\documentclass[11pt,a4paper]{article}
\usepackage{geometry}
\usepackage{graphicx}
\usepackage{bm,array}
\usepackage{comment}
\usepackage{caption}
\usepackage{subcaption}
 \usepackage{makecell}
 \usepackage{multirow}
 \usepackage{float}

\usepackage[normalem]{ulem}

\usepackage[pdfpagemode={UseOutlines},bookmarks=true,bookmarksopen=true,
bookmarksopenlevel=0,bookmarksnumbered=true,hypertexnames=false,
colorlinks,linkcolor={blue},citecolor={blue},urlcolor={red},
pdfstartview={FitV},unicode,breaklinks=true]{hyperref}
\hypersetup{urlcolor=blue, colorlinks=true}

\usepackage{setspace}
\usepackage{vmargin}
\setmarginsrb{ 1.0in}  
{ 0.6in}  
{ 1.0in}  
{ 0.8in}  
{  20pt}  
{0.25in}  
{9pt}  
{ 0.3in}  

\usepackage{authblk}

\usepackage{amsmath}
\usepackage{amsfonts}
\usepackage{amssymb}
\usepackage[round, sort,comma,authoryear]{natbib}

\newtheorem{theorem}{Theorem}

\newtheorem{corollary}{Corollary}

\newtheorem{definition}[theorem]{Definition}

\newtheorem{proposition}{Proposition}

\newenvironment{proof}[1][Proof]{\noindent\textbf{#1.} }{\ \rule{0.5em}{0.5em}}

\usepackage{multirow}
\usepackage{hhline}
\usepackage{caption}
\usepackage{subcaption}
\usepackage{footnote}
\usepackage[ruled,vlined]{algorithm2e}
\usepackage{algorithmic}
\usepackage{soul}

\newcommand{\cS}{{\mathcal{S}}}

\newcommand{\cL}{{\mathcal{L}}}

\newcommand{\beps}{{\pmb{\epsilon}}}

\newcommand{\bbE}{\mathbb{E}}

\usepackage{mathtools}

\newcommand{\RL}{\text{\tiny RL}}

\newcommand{\CRL}{\text{\tiny CRL}}
\newcommand{\ERL}{\text{\tiny ERL}}

\usepackage{mathtools}

\newif\ifnotes\notestrue
\def\htien#1{}

\begin{document}






\newcolumntype{C}{>{\centering\arraybackslash}p{4em}}

\title{\textbf{Constrained Recursive Logit for Route Choice Analysis}}
\author[1]{Hung Tran}
\author[1,*]{Tien Mai}
\author[2]{Minh Ha Hoang}
\affil[1]{\it\small
School of Computing and Information Systems, Singapore Management University}
\affil[2]{\it\small
SLSCM and CADA, Faculty of Data Science and Artificial Intelligence, College of Technology, National Economics University, Hanoi, Vietnam}
\affil[*]{\it\small
Corresponding author, atmai@smu.edu.sg}

\maketitle
\begin{abstract}
The recursive logit (RL) model has become a widely used framework for route 
choice modeling, but it suffers from a key limitation: it assigns non-zero 
probabilities to all paths in the network, including those that are unrealistic, 
such as routes exceeding travel-time deadlines or violating energy constraints. 
To address this gap, we propose a novel \emph{Constrained Recursive Logit (CRL)} model 
that explicitly incorporates feasibility constraints into the RL framework. CRL 
retains the main advantages of RL---no path sampling and ease of prediction---but 
systematically excludes infeasible paths from the universal choice set.  
The model is inherently non-Markovian; to address this, we develop a tractable 
estimation approach based on extending the state space, which restores the 
Markov property and enables estimation using standard value iteration methods.
 We prove that our estimation method admits a 
unique solution under positive discrete costs and establish its equivalence to a 
multinomial logit model defined over restricted universal path choice sets. Empirical experiments 
on synthetic and real networks demonstrate that CRL improves behavioral realism 
and estimation stability, particularly in cyclic networks. 
\end{abstract}







%


\section{Introduction}\label{sec:intro}
Route choice modeling is a central component of travel demand analysis, providing 
the basis for applications ranging from network design and traffic management to 
policy evaluation and infrastructure pricing. Among the existing frameworks, the 
recursive logit (RL) model \citep{FosgFrejKarl13} has become particularly popular 
due to its computational advantages: it avoids sampling of alternative routes, can 
be consistently estimated using maximum likelihood methods, and offers a tractable 
structure for prediction through dynamic programming.  

Despite these strengths, a fundamental limitation of the RL model lies in its 
assumption of a \emph{universal choice set}, where all possible paths between an 
origin and a destination are included in the consideration choice sets. This 
assumption often results in the assignment of positive probabilities to paths that 
are clearly unrealistic, such as those with excessively long travel times or those 
that violate natural feasibility conditions (e.g., exceeding an energy budget or 
ignoring temporal deadlines). As a consequence, the RL model may produce biased 
parameter estimates and misleading behavioral predictions in contexts where 
constraints play a defining role. To date, this limitation has not been 
satisfactorily addressed in the literature, leaving a gap between theory and 
practical travel behavior.  

In this paper, we propose a novel \emph{Constrained Recursive 
Logit (CRL)} model that explicitly incorporates feasibility constraints into the 
recursive logit framework. The proposed CRL model retains the computational 
advantages of the classical RL formulation---it requires no sampling of path 
alternatives and remains straightforward to predict---while systematically 
eliminating infeasible paths from the universal choice set. By doing so, the CRL 
model ensures that estimated choice probabilities reflect only realistic and 
feasible route choices.  
Specifically, we make the following contributions:
\begin{itemize}
    \item \textbf{Constrained RL formulation.}  
    We propose a CRL framework that explicitly 
    incorporates feasibility constraints into the RL setting. In particular, we 
    enforce that any path whose accumulated cost at some state exceeds a 
    prespecified threshold is assigned zero choice probability and is excluded 
    from the universal choice set. This formulation is motivated by several 
    practical applications, such as route choice with travel-time deadlines or 
    electric vehicle routing with battery capacity limits. We further show that, 
    under these constraints, the induced choice probabilities are inherently 
    non-Markovian, as they depend on the full sequence of previously visited 
    states.  

    \item \textbf{Tractable estimation via state-space extension.}  
    Since the non-Markovian property renders direct estimation infeasible, we 
    propose a novel approach based on extending the state space to include the 
    accumulated cost. This transformation restores the Markovian property and 
    enables estimation using standard value iteration and the Nested Fixed-Point 
    (NFXP) algorithm \citep{Rust87}. Importantly, we prove that when costs are positive and 
    discrete, the resulting linear system for computing the expected value 
    function always admits a unique solution, thereby guaranteeing the existence 
    and stability of the estimation procedure.  

    \item \textbf{Equivalence to path-based multinomial logit (MNL) under restricted universal choice sets.}  
    We establish the equivalence between the CRL formulation and a standard 
    path-based MNL model defined over the restricted 
    universal choice set of feasible paths. Here we note that while constraints could in principle 
    be handled in a path-based MNL model by sampling paths and discarding 
    infeasible ones, such an approach suffers from the same drawbacks as 
    path-sampling methods in general: incomplete coverage of observed paths and 
    limited predictive capability \citep{FosgFrejKarl13,zimmermann2020tutorial}. By contrast, CRL eliminates infeasible paths 
    automatically without requiring path sampling.  

    \item \textbf{Extensions to multiple constraints and nested RL.}  
    We demonstrate that the framework is general and flexible. In particular, we 
    extend the CRL formulation to accommodate multiple simultaneous constraints 
    (e.g., combining travel-time and energy limits) and show how it can be 
    integrated into the nested recursive logit (NRL) model \citep{MaiFosFre15} to capture correlation 
    structures across the transportation network.  

    \item \textbf{Empirical validation.}  
    We conduct extensive experiments on both synthetic and real transportation 
    networks. The results confirm that CRL provides more behaviorally realistic 
    predictions than the standard RL model when constraints are present, and that 
    its estimation procedure remains stable especially in highly cyclic networks where 
    RL estimation typically fails.  
\end{itemize}

Our CRL model might have wide-ranging applications in constrained travel settings. 
Examples include trip planning under deadlines, electric vehicle routing where 
battery capacity limits the feasible path set, and shared mobility systems such as 
bike-sharing, where trip duration or rental constraints restrict traveler 
decisions. More generally, the CRL model bridges the gap between the theoretical 
elegance of recursive logit and the practical realities of constrained travel 
behavior, offering a robust and versatile tool for both empirical estimation and 
policy analysis in modern transportation systems.  

Although our discussions in this paper  focus on the route choice problems, the proposed CRL 
framework has broader applicability. In particular, it provides a general approach 
for incorporating feasibility constraints into dynamic discrete choice settings, 
which are widely used in fields such as activity-based travel demand modeling and 
sequential decision processes. For example, activity-based models often require 
capturing temporal or resource constraints in daily activity schedules 
\citep{BowmanBenAkiva2001}. Similarly, in dynamic discrete choice models 
(DDCMs), feasibility conditions such as budget limits, time windows, or capacity 
restrictions play an important role in determining choice behavior 
\citep{Rust1987, AguirregabiriaMira2010}. By systematically excluding infeasible 
alternatives from the universal choice set, CRL can offer a unified and tractable way 
to enhance behavioral realism across these domains.

The remainder of the paper is structured as follows. 
Section~\ref{sec:review} provides a brief review of the related literature. 
Section~\ref{sec:RL} presents the standard RL model, which 
serves as the foundation for our proposed CRL framework. 
Section~\ref{sec:CRL} introduces the CRL formulation together with illustrative 
examples, while Section~\ref{sec:model-properties} explores the fundamental 
properties of the model. Estimation methods based on the extended state-space 
representation are discussed in Section~\ref{sec:extended-state-space}. 
Section~\ref{sec:exp} reports the results of numerical experiments, and 
Section~\ref{sec:concl} concludes the paper. 
The Appendix provides further details, including extensions of CRL to handle 
multiple constraints and nested RL.


\section{Literature Review}\label{sec:review}
The literature on route choice modeling can broadly be categorized into 
path-based and link-based recursive approaches. Path-based models 
\citep[see][for a review]{Prat09} rely on sampling of paths, which makes the 
resulting parameter estimates sensitive to the sampling procedure. Even when 
correction terms are introduced to ensure consistent estimation, prediction 
remains computationally challenging, particularly in large networks.  

In contrast, link-based recursive models \citep{FosgFrejKarl13,MaiFosFre15,mai2021RL_STD,oyama2017discounted} build on the dynamic discrete choice 
framework of \citet{Rust87} and are essentially equivalent to discrete choice 
models defined over the set of all feasible paths. These models have several 
attractive advantages: they can be consistently estimated without requiring 
path sampling, and they allow efficient prediction through dynamic programming. 
The first RL model was introduced by \citet{FosgFrejKarl13}, 
and has since been extended in various directions, such as accounting for 
correlations between path utilities 
\citep{MaiFosFre15,Mai_RNMEV,MaiBasFre15_DeC}, handling dynamic networks 
\citep{de2020RL_dynamic}, incorporating stochastic time-dependent link costs 
\citep{mai2021RL_STD}, or modeling discounted behavior \citep{oyama2017discounted}. 
Applications of recursive models can be found in traffic management 
\citep{BailComi08,Melo12}, network pricing \citep{zimmermann2021strategic}, network interdiction \citep{mai2024stackelberg}, and 
beyond. A comprehensive review is provided in \citet{zimmermann2020tutorial}.  

Despite their many strengths, most existing RL models share two fundamental 
limitations. First, they do not impose feasibility constraints, which implies 
that \emph{any} path in the network is assigned a non-zero probability, even if 
it is clearly unrealistic (e.g., excessively long travel times or violations of 
energy or temporal constraints). Second, RL is known to suffer from instability 
in estimation when applied to highly cyclic networks, which is typically the 
case in real-world transportation systems \citep{MaiFrejinger22}. These issues limit the behavioral 
realism and empirical robustness of RL models.  
Our proposed CRL framework retains the computational and 
estimation advantages of the classical RL model, but introduces feasibility 
constraints directly into the recursive formulation. In doing so, CRL ensures 
that infeasible routes are systematically excluded from the universal choice 
set, thereby resolving both the issue of unrealistic path probabilities and the 
instability of estimation in cyclic networks. This makes CRL a more robust and 
behaviorally consistent foundation for route choice modeling in constrained 
transportation systems.

It is worth noting that our CRL model is closely related to  
\emph{prism-based approaches}, which also extend RL models by restricting the 
universal path set. Specifically, \citet{OyaHat19} introduced a prism-based path 
set restriction for Markovian traffic assignment, pruning paths outside a 
feasible travel-time prism and thereby avoiding the assignment of positive 
probabilities to unrealistically long routes. Building on this idea, 
\citet{oyama2023prism_estimation} applied prism constraints to RL estimation, 
showing that they improve stability and better capture positive network 
attributes by excluding infeasible alternatives. While prism-based methods incorporate path feasibility, they are limited to 
unit-value constraints (i.e., assigning a cost of one per transition). In contrast, 
our CRL framework generalizes this idea by accommodating a broader class of 
constraints---such as travel-time deadlines, energy budgets, or rental duration 
limits---directly within the recursive structure. This generalization yields a 
unified and flexible framework for modeling constrained route choice behavior.

Our CRL framework is also related to the literature on constrained discrete choice 
models, which has primarily focused on introducing feasibility restrictions within 
classical static frameworks. For example, the \emph{constrained nested logit model} 
\citep{constrained_nested_logit} incorporates both hard constraints, which eliminate 
infeasible alternatives, and soft constraints, which reduce the attractiveness of 
certain options. Likewise, spatially constrained destination choice models 
\citep{constrained_destination_choice} impose restrictions based on land use and trip 
distribution, thereby yielding more behaviorally realistic predictions. While these 
approaches highlight the value of incorporating constraints into discrete choice 
models, they remain static and path-based, and therefore differ fundamentally from 
link-based recursive formulations such as CRL in both their estimation strategies and 
predictive capabilities.  

The RL and CRL models also share a strong connection with the literature on dynamic 
discrete choice models (DDCMs) \citep{Rust87,aguirregabiria2010ddcm,mai2020relation}. Recent 
extensions of DDCMs \citep{bruneel2025dynamic,Chen2025ModelAdaptive} have advanced 
computational techniques for handling large state spaces and, in some cases, allow 
for implicit incorporation of constraints. However, these models are predominantly 
applied in economics and industrial organization rather than in transportation 
networks. By contrast, our work focuses specifically on 
transportation networks, and establishes a direct connection between the link-based 
RL framework and the path-based MNL model under both universal 
and restricted universal choice sets. In this sense, the CRL model bridges the gap 
between the constrained discrete choice literature and recursive route choice 
modeling, providing both behavioral realism and computational tractability.

\section{Background - The Recursive Logit Model}\label{sec:RL}

Consider a set of states $\cS$, where each state $s \in \cS$ represents a link or node in the transportation network. Let $N(s) \subseteq \cS$ denote the set of states that are directly reachable from state $s$. Here, $N(s)$ can represent the set of outgoing nodes (or links) from state $s$. For each trip with destination $d$, we introduce $d$ as an absorbing state into the system, making the set of all states $\widetilde{\cS} = \cS \cup \{d\}$. 

In the context of the recursive logit (RL) model, each path choice is modeled as a sequence of state choices, which could correspond to a sequence of link or node choices in the transportation network. Under the random utility maximization framework, each pair of states $(s, s')$, where $s' \in N(s)$, is associated with an instantaneous random utility:
\[
u(s'|s) = v(s'|s) + \mu \epsilon(s'),
\]
where:
\begin{itemize}
    \item $v(s'|s)$ is the deterministic utility, representing the traveler’s perceived utility of moving from state $s$ to state $s'$,
    \item $\epsilon(s')$ is a random term, assumed to be i.i.d. extreme value type I,
    \item $\mu$ is the scale parameter.
\end{itemize}
At each state, the traveler is assumed to choose the next state to maximize their expected utility:
\[
\bbE_{\epsilon} \left[ \max_{s' \in N(s)} \{ v(s'|s) + V(s') + \mu \epsilon(s') \} \right],
\]
where $V(s)$ is the expected maximum utility from state $s$ to the destination $d$. This expected utility function can be recursively calculated as:
\[
V(s) =
\begin{cases}
    0, & s = d, \\
    \mu \ln\left(\sum_{s' \in N(s)} e^{\frac{1}{\mu}(v(s'|s) + V(s'))}\right), & \forall s \in \cS.
\end{cases}
\]

Here we note that the value function $V(s)$ could be  dependent of both the destination $d$ and the traveler's characteristics, but  these indicators are omitted here for notational simplicity.

Given the value function $V$, the probability that a traveler chooses state $s'$ from state $s$ can be computed as:
\[
P(s'|s) =
\begin{cases}
    \frac{\exp(v(s'|s) + V(s'))}{\sum_{t \in N(s)} \exp(v(t|s) + V(t))}, & \text{if } s' \in N(s), \\
    0, & \text{otherwise}.
\end{cases}
\]
The probability of observing a path $\sigma = \{s_0, \ldots, s_n\}$ can then be calculated as:
\[
P(\sigma) = \prod_{t=0}^{n-1} P(s_{t+1}|s_t) = \exp(v(\sigma) - V(s_0)),
\]
where $v(\sigma)$ is the total deterministic utility of the path $\sigma$, defined as:
$v(\sigma) = \sum_{t=0}^{n-1} v(s_{t+1}|s_t)$.

For model estimation, each deterministic utility function \(v(s'|s)\) can be specified as a (linear) function of various link or node attributes, such as travel time and cost, along with parameters to be estimated. Given a set of observed paths, the model parameters can be estimated using Maximum Likelihood Estimation (MLE) by maximizing the log-likelihood of these observations \citep{Mcfadden2001economicNobel,FosgFrejKarl13}.

It is important to note that the RL model can be considered a specific case of the logit-based dynamic discrete choice model \citep{aguirregabiria2010ddcm,mai2020relation,Rust87}. To elaborate, let \(\phi_t(s_t | \beps_t): \cS \rightarrow \cS\) represent a decision rule at time step \(t = 0, 1, \ldots\), determining the next state \(s_{t+1}\) based on the current state \(s_t\) and the realization of the random terms \(\beps_t\). The RL model can then be framed as a logit-based discrete choice model, where the objective is to select a set of decision rules that maximize the expected long-term utility:
\[
\max_{\phi_0,\phi_1,\ldots} \bbE_{\beps}\left[\sum_{t = 0}^\infty v(s_{t+1}|s_t) +\mu \epsilon(s_{t+1})\right].
\]


\section{Constrained Recursive Logit Model}\label{sec:CRL}
In the following, we present the formal formulation of the CRL model and provide 
illustrative examples that demonstrate how it captures constrained route choice 
behavior in networks.

\subsection{ Modeling Formulation}
The CRL model extends the standard RL framework by explicitly incorporating 
constraints into travelers’ route choice behavior.
 Specifically, we assume that each movement from state \(s_t\) to state \(s_{t+1}\) is associated with a cost function \(c(s_{t+1} \mid s_t)\), reflecting the cost incurred by travelers when transitioning from \(s_t\) to \(s_{t+1}\). These cost values can take either positive or negative values. Travelers are assumed to make route choice decisions by maximizing their expected accumulated random utility while ensuring that the total accumulated cost at every point in time does not exceed a given threshold $\alpha$.  

For example, in the context of electric vehicles (EVs), drivers must ensure that, at any point in time,  the total consumed energy does not exceed the remaining battery charge before reaching a charging station or their destination. Here, \(c(s_{t+1}|s_t)\) represents the energy consumed on the path from \(s_t\) to \(s_{t+1}\), which depends on factors like travel time and road conditions, while \(\alpha\) corresponds to the initial battery charge at the start of the trip.  The cost function \(c(s_{t+1}|s_t)\) typically takes positive values, reflecting energy consumption. However, it can take negative values if the vehicle visits a charging station, where it can partially \textit{recharge} or \textit{replace the battery}. This flexibility allows the model to realistically capture travel scenarios involving resource constraints. Another example is the context of bike-sharing systems, where companies often impose a time limit on each rental session. For instance, a bike may need to be returned to a designated station within a specific duration, such as 30 minutes, to avoid additional fees or penalties. For example, in Singapore, Anywheel's daily pass allows users unlimited rides per day, but each trip cannot exceed 30 minutes or they incur a surcharge of S\$0.50 per 30 minutes (\url{https://www.sharedmobility.news/bike-sharing-time-crack-free-riding}). Such a requirement forces travelers to plan their routes carefully to ensure they can return the bike within the allowed time frame. If the total trip exceeds the time budget, the traveler may need to return the bike to a station, complete the current rental session, and start a new one to continue their journey. 

\paragraph{Non-Markovian Property.} It can be observed that a policy solving  the above constrained modeling problem would be \textit{non-Markovian }(and therefore not \textit{stationary}). This means that the choice probability at a state $s_t$ would depend not only on the state $s_t$ itself but also on the entire historical trajectory $\{s_0, \ldots, s_t\}$. 
To illustrate this, consider the simple example shown in Figure \ref{fig:example_figure}. The network contains four nodes, where $a$ is the origin and $d$ is the destination. Each edge is labeled with its associated cost. Let $\alpha = 12$, which implies that we need to  seek a policy assigning zero probability to paths whose cost exceeds 12.  In this scenario, from node $c$, if the traveler takes the edge $[a,c]$, both outgoing edges from $c$ to $d$ will form paths with a total cost less than 12. Thus, both edges $[c,4,d]$ and $[c,1,d]$ are feasible choices. On the other hand, if the traveler takes the path $[a,b,c]$, they cannot take the edge $[c,4,d]$ because the total cost would exceed $\alpha$. In this case, $[c,1,d]$ becomes the only feasible path to reach the destination.

From this example, it is evident that the policy at state $c$ depends on the historical path the traveler took to arrive at $c$. This demonstrates the non-Markovian property of the constrained policy. This phenomenon can also be observed in a more general context. For instance, in an EV scenario, at a state, if the driver chooses a path that helps conserve energy, they might be able to reach the destination directly without stopping at a charging station. In contrast, if the driver has depleted energy due to a historically long or energy-intensive route, they will need to stop at a charging station to ensure they can reach their destination.

\begin{figure}[h!]
    \centering
    \includegraphics[width=0.7\linewidth]{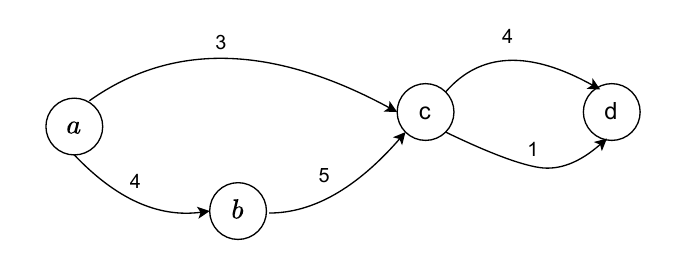}
    \caption{Example to illustrate the non-Markovian property of the CRL.}
    \label{fig:example_figure}
\end{figure}

 \paragraph{Mathematical Formulation. }
To model such a non-Markovian behavior, let \(\pi(s_{t+1} \mid S_t)\) denote the probability that the traveler chooses the next state \(s_{t+1}\) conditional on a joint state \(S_t = \{s_0, \ldots, s_t\}\), which represents the historical path from the origin to state \(s_t\). Due to the incorporation of constraints, the resulting decision process is no longer Markovian—the policy at state \(s_t\) may depend on the entire trajectory up to \(s_t\). We will discuss the non-Markovian property in more detail in the next section.

At each state \(s_t\), the traveler chooses the next state \(s_{t+1}\) by maximizing the following expected utility:
\[
\mathbb{E}_{\epsilon} \left[ \max_{s_{t+1} \in N(s_t)} \left\{ v(s_{t+1} \mid s_t) + V(S_{t+1}) + \mu \epsilon(s_{t+1}) \right\} \right],
\]
where \(v(s_{t+1} \mid s_t)\) denotes the immediate deterministic utility of moving from \(s_t\) to \(s_{t+1}\), \(V(S_{t+1})\) is the expected utility from the joint state \(S_{t+1} = S_t \cup \{s_{t+1}\}\) to the destination, and \(\epsilon(s_{t+1})\) are i.i.d. Gumbel (Type I extreme value) distributed random utilities.

Under this assumption, the value function \(V(S_t)\) from the joint state \(S_t\) is computed as:
\begin{equation}\label{eq:V(S-t)}
   V(S_t) =
\begin{cases}
    0, & \text{if } s_t = d \text{ and } c(S_t)\leq \alpha\\
    -\infty, & \text{if } C(S_t) > \alpha, \\
    \mu \ln \left( \sum_{s_{t+1} \in N(s_t)} \exp\left( \frac{1}{\mu} \left[ v(s_{t+1} \mid s_t) + V(S_{t+1}) \right] \right) \right), & \text{otherwise},
\end{cases} 
\end{equation}
where \(C(S_t)\) is the accumulated cost along the path \(S_t\), given by
\[
C(S_t) = \sum_{i=0}^{t-1} c(s_{i+1} \mid s_i),
\]
and \(\alpha\) is the upper bound on the allowable accumulated cost. Intuitively, when the cost exceeds this threshold, the value function becomes \(-\infty\), ensuring that the corresponding choice probability becomes zero.

Given the value function, the choice probability is given by a  MNL model:
\[
\pi(s_{t+1} \mid S_t) = 
\begin{cases}
    \frac{\exp\left(\frac{1}{\mu} (v(s_{t+1} \mid s_t) + V(S_{t+1})) \right)}{\sum_{s \in N(s_t)} \exp\left( \frac{1}{\mu}(v(s \mid s_t) + V(S_t \cup \{s\})) \right)}, & \text{if } s_{t+1} \in N(s_t), \\
    0, & \text{otherwise}.
\end{cases}
\]
Finally, given an observed path \(\sigma = \{s_0, \ldots, s_T\}\), the path probability under the constrained RL model is given by:
\[
P(\sigma) = \prod_{t=0}^{T-1} \pi(s_{t+1} \mid S_t).
\]

It can be seen that, under the value function defined above, at any state \( S_t \), if there exist two candidate next states \( s \) and \( s' \) such that the accumulated cost of transitioning to \( s \) remains below the threshold \(\alpha\), while the accumulated cost of transitioning to \( s' \) exceeds \(\alpha\), then—as expected—the probability of choosing state \( s' \) will be zero. We formalize this observation in the following proposition:
\begin{proposition}
Let \( S_t = \{s_0, s_1, \ldots, s_t\} \) denote a joint state (i.e., the historical path up to time \( t \)), and let the accumulated cost along this path be \( C(S_t) = \sum_{i=0}^{t-1} c(s_{i+1} \mid s_i) \). Suppose there exists at least one feasible next state \( s \in \mathcal{S} \) such that the total cost of transitioning from \( s_t \) to \( s \) remains within the allowed threshold, i.e., \( C(S_t) + c(s \mid s_t) \leq \alpha \). Then, for any alternative next state \( s' \in \mathcal{S} \) such that the accumulated cost violates the threshold, i.e., \( C(S_t) + c(s' \mid s_t) > \alpha \), the probability of choosing \( s' \) is zero, that is, \( \pi(s' \mid S_t) = 0 \). As a consequence, for any full path \( \sigma = \{s_0, s_1, \ldots, s_T\} \), if there exists a sub-path \( \sigma' = \{s_0, \ldots, s_H\} \) for some \( H \leq T \) such that the accumulated cost along \( \sigma' \) exceeds the threshold \( \alpha \), i.e., \( \sum_{t=0}^{H-1} c(s_{t+1} \mid s_t) > \alpha \), then the probability of  the entire path \( \sigma \) is zero, i.e., \( P(\sigma) = 0 \).
\end{proposition}
\begin{proof}
The proof is straightforward. From the recursive formulation of the value 
function $V(S_t)$, consider a state $s_t$. If there exists a successor state 
$s'$ such that 
$C(S_t) + c(s' \mid s_t) > \alpha,$
then the corresponding value function becomes 
$V(S_t \cup \{s'\}) = -\infty$. Consequently, the choice probability formulation 
in~\eqref{eq:V(S-t)} implies that 
$\pi(s' \mid S_t) = 0.$

Now, consider a path $\sigma = \{s_0, s_1, \ldots, s_T\}$. Suppose there exists 
a subpath $\sigma' = \{s_0, \ldots, s_H\}$ for some $H \leq T$ such that the 
accumulated cost along $\sigma'$ exceeds the threshold $\alpha$, i.e.,
$\sum_{t=0}^{H-1} c(s_{t+1} \mid s_t) > \alpha.$
In this case, there must exist an extended state $\widetilde{S}_i = \{s_0,...,s_i\}$ for some 
$i \in \{0,\ldots,H-1\}$ such that $\pi(s_{i+1} \mid \widetilde{S}_i) = 0$. This implies 
that the probability of selecting the entire path $\sigma$ is zero, i.e.
$P(\sigma) = 0,$
which establishes the claim.
\end{proof}

The above proposition highlights an important implication of the model: for any path, if at any point along that path the accumulated cost exceeds the specified upper bound \( \alpha \), then the probability of taking that path is exactly zero. This result is consistent with the real-world constraints observed in several application domains. For instance, in the EV routing example, a driver must ensure that the battery never depletes entirely along the journey. As such, the driver would never consider a path that would result in the energy level dropping to zero at any point along the way, which corresponds to a constraint violation and thus leads to zero path probability. Similarly, in the context of bike-sharing or bike rental systems, consider a policy where users are required to return the bike within 30 minutes before initiating a new rental; otherwise, they incur a substantial penalty \citep{Tsushima2023timelimit}.  In this case, a traveler (or renter) would not consider any path that violates this time constraint at any moment along the route. The model naturally encodes such behavior: any path that violates the constraint, even temporarily, will be assigned zero probability and effectively ruled out from the consideration choice set of paths.

\subsection{Illustrative Examples}

We present some simple examples to illustrate the implications of the  modeling framework discussed above.

\paragraph{Example 1: Route choice  with travel time upper-bounds}
The first scenario considers a route choice problem in which we assume that travelers will not consider invalid route whose travel times  exceed a travel time upper-bound $\alpha$. This is a common case in real world: A taxi driver chooses a route to pick up customer on time, or a shipper plans to deliver packages before a deadline.

\begin{figure}[htbp]
    \centering
    \includegraphics[width=0.5\textwidth]{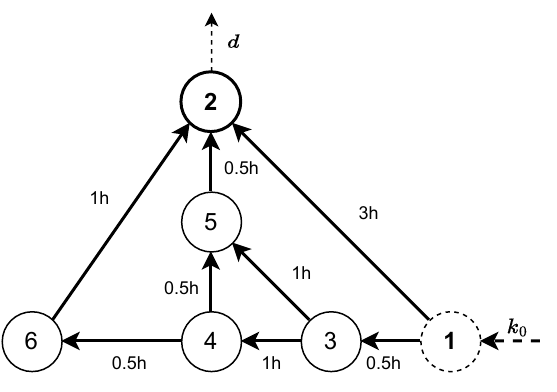} 
    \caption{Toy network with travel times}
    \label{fig:toy-net-1.1} 
\end{figure}

We consider a toy network shown in Figure~\ref{fig:toy-net-1.1}, where node \( k_0 \) is the origin and node \( d \) is the destination. The number on each edge represents its travel time. The utility associated with each transition from node \( s \) to node \( s' \) is computed as \( v(s' \mid s) = -2 \times TT(s' \mid s) \), where \( TT(s' \mid s) \) denotes the travel time from \( s \) to \( s' \). The cost is also defined as the travel time on each edge, and we impose an upper bound \( \alpha = 2.5 \), meaning that any path with a total travel time exceeding 2.5 hours will be assigned zero probability under the constrained model.

In Table~\ref{tab:path-choice-constraints}, we report the path probabilities computed under both the original (unconstrained) RL model and the constrained RL model with \( \alpha = 2.5 \). The values in parentheses represent the total travel time for each path. As observed, in the unconstrained setting, all paths have non-zero probabilities. In contrast, when the constraint is enforced, the probabilities of paths \([1,2]\) and \([1,3,4,6,2]\) become zero due to violation of the travel time limit.


\begin{table}[htb]
\centering
\begin{minipage}{0.45\linewidth}
\centering
\begin{tabular}{l|l|l}
No. & Path & $P(\sigma^n)$ \\\hline
1   & [1,2]& 0.083 \\
2   & [1,3,5,2]& 0.610 \\
3   & [1,3,4,5,2]& 0.224 \\
4   & [1,3,4,6,2]& 0.083
\end{tabular}
\caption*{(a) Path probabilities given by the RL model without constraints}
\end{minipage}
\hfill
\begin{minipage}{0.45\linewidth}
\centering
\begin{tabular}{l|l|l}
No. & Path & $P(\sigma^n)$ \\\hline
1   & [1,2] (3h)                        & 0.000 \\
2   & [1,3,5,2] (2h)          & 0.731 \\
3   & [1,3,4,5,2] (2.5h)   & 0.269 \\
4   & [1,3,4,6,2] (3h)   & 0.000
\end{tabular}
\caption*{(b) Path probabilities given by the constrained RL model}
\end{minipage}
\caption{Path choice probabilities given by the RL and constrained RL models.}
\label{tab:path-choice-constraints}
\end{table}

\subsection{Example 2: Route choice analysis for rechargeable vehicles}
In the previous example, the costs represent travel times, which are always non-negative. In this section, we present a scenario where the cost can take both negative and positive values. Specifically, we consider a route choice problem for rechargeable vehicles that consume some form of energy (e.g., fuel or electricity) to operate. During long trips, such vehicles must periodically refuel or recharge at gas stations or charging stations.

In this scenario, certain locations in the network represent charging stations where a vehicle can replenish its energy up to a specified level in order to continue the trip. The cost function \( c(s' \mid s) \) represents the net energy change when transitioning from state \( s \) to state \( s' \). Specifically, this cost takes a positive value when the vehicle travels along a road segment—corresponding to energy consumption—and a negative value when the vehicle visits a charging station—corresponding to energy replenishment.
At the origin, the initial negative cost reflects the vehicle's starting energy level. The energy constraint imposes that the accumulated cost at any point in time must remain less than or equal to the vehicle's maximum energy capacity, denoted by \( \alpha \). Formally, this is expressed as \( C(S_t) \leq \alpha \) for all time steps \( t \). This condition ensures that the vehicle does not exceed its energy capacity and has sufficient energy to complete its journey.

\begin{figure}[htb] 
    \centering
    \includegraphics[width=0.6\linewidth]{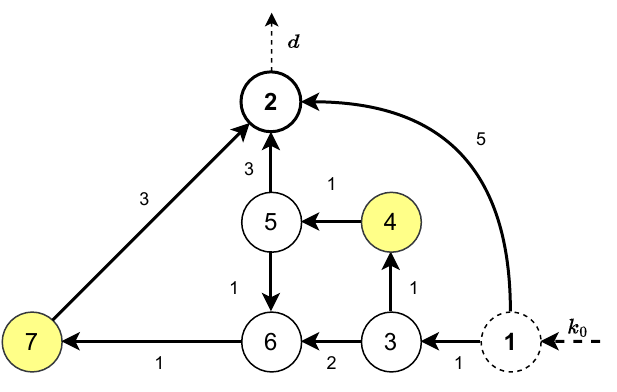}
    \caption{Toy network with charging stations.} 
    \label{fig:energy} 
\end{figure}

We consider the toy network illustrated in Figure~\ref{fig:energy}, which comprises 7 nodes and 10 links. The driver seeks to travel from the origin node \( k_0 \) to the destination node \( d \). Each link is annotated with its travel time, and for simplicity, we assume that the travel time on a link is equal to the amount of energy consumed when traversing it. In this setting, Nodes 4 and 7 are designated as recharging stations. When the vehicle visits either of these nodes, the accumulated cost—representing the total energy consumed—is reset to zero, simulating a full recharge. We model the energy constraint by setting the initial energy level to zero, representing a fully charged vehicle at the start of the trip. Each traversal of a road link increases the accumulated cost (i.e., depletes energy), while visits to recharging stations reset this cost to zero, thereby allowing the vehicle to continue its journey without violating the energy constraint.

\begin{table}[htb]
    \centering
    \begin{tabular}{l|l|l|l|l|l}
No. & Path              & $P(\sigma)$ ($\alpha = 5$) &$P(\sigma)$ ($\alpha = 4$) & $P(\sigma)$ ($\alpha = 3$)                  \\\hline
1   & [1,2]                & 0.644 &0&0 \\
2   & [1,3,4,5,2]        & 0.237&0.665&0 \\
3   & [1,3,4,5,6,7,2]    & 0.032&0.090&1.000 \\
4   & [1,3,6,7,2]    & 0.087&0.245&0
\end{tabular}
    \caption{Path choice probabilities under different values of $\alpha$.}
    \label{tab:energy-5}
\end{table}

\begin{figure}[htb] 
\centering
    \begin{subfigure}{0.45\linewidth}
         \centering
         \includegraphics[width=\textwidth]{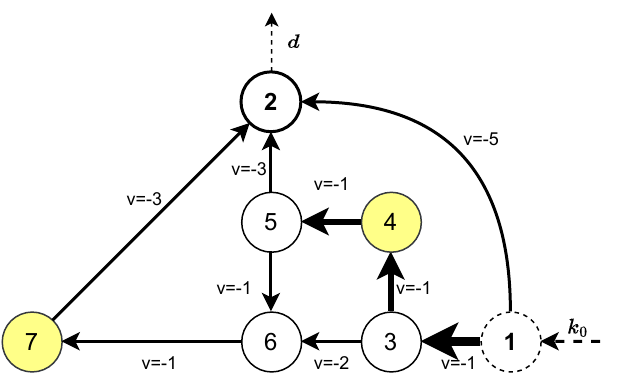}
         \caption{$\alpha=5$}\label{subfig:energy-5}
    \end{subfigure}
    \begin{subfigure}{0.45\linewidth}
         \centering
         \includegraphics[width=\textwidth]{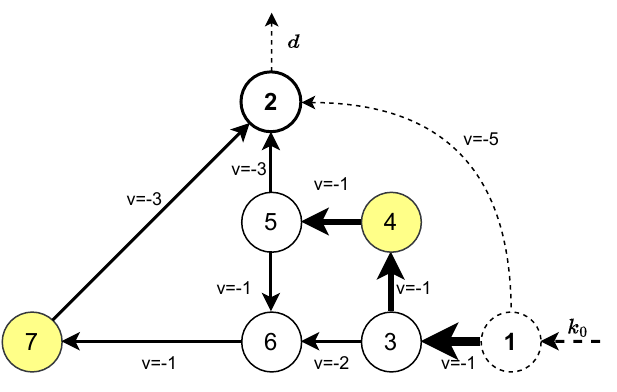}
         \caption{$\alpha=4$}\label{subfig:energy-4}
    \end{subfigure}
    \begin{subfigure}{0.45\linewidth}
         \centering
         \includegraphics[width=\textwidth]{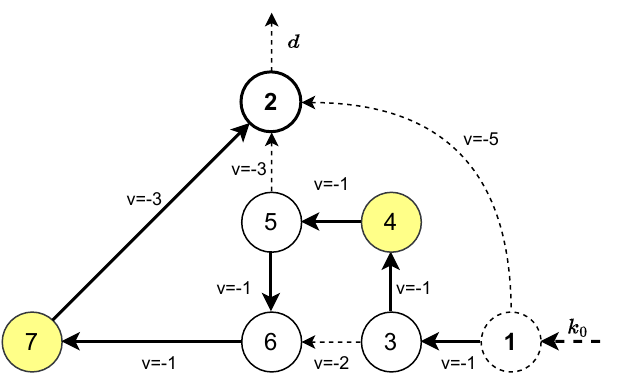}
         \caption{$\alpha=3$}\label{subfig:energy-3}
    \end{subfigure}
    \caption{Visualization of path choice probabilities under varying energy thresholds \( \alpha \).} 
    \label{fig:energy-constrained} 
\end{figure}





Different values of the constraint parameter \( \alpha \) lead to different sets of feasible paths. Intuitively, when \( \alpha \) is high—corresponding to a vehicle with greater energy capacity—more paths become feasible and are included in the route choice set. In contrast, when \( \alpha \) is low, the vehicle must visit charging stations more frequently to ensure it has sufficient energy to complete the trip. In the following, we illustrate how path choice probabilities vary with different values of \( \alpha \), specifically \( \alpha \in \{3, 4, 5\} \).

In Table~\ref{tab:energy-5}, we report the path choice probabilities computed using the constrained Recursive Logit (RL) model under different energy thresholds \( \alpha \). As observed, when \( \alpha = 5 \), all candidate paths satisfy the energy constraint and are therefore feasible, resulting in strictly positive choice probabilities for all paths. However, as the threshold \( \alpha \) decreases, some paths become infeasible due to excessive energy consumption without sufficient opportunity for recharging. Consequently, these paths receive zero probability under the constrained model. In particular, when \( \alpha = 3 \), only a single path—namely \([1,3,4,5,6,7,2]\)—remains feasible. This path includes both designated charging stations at Nodes 4 and 7, allowing the vehicle to replenish its energy and continue the trip under the  energy constraint.

For visualization purposes, Figure~\ref{fig:energy-constrained} illustrates the path choice probabilities under three different values of the energy constraint parameter \( \alpha \). The link utilities \( v(s' \mid s) \) are defined as \( -2 \times TT(s' \mid s) \), where \( TT(s' \mid s) \) is the travel time between nodes \( s \) and \( s' \). In the figure, dashed links represent transitions with zero probabilities under the constrained model, while the thickness of each solid link reflects the magnitude of its corresponding choice probability. For example, when \( \alpha = 5 \), all links are feasible and thus have positive probabilities. Among them, the link \( (1,3) \) has the highest probability of being chosen. Furthermore, at Node 3, the probability of transitioning to Node 4 (a charging station) is higher than the probability of transitioning to Node 6. In contrast, when \( \alpha = 3 \), several links—including \( (1,2) \), \( (3,6) \), and \( (5,2) \)—have zero probability of being selected, as they lead to infeasible paths that violate the energy constraint. This visualization highlights how stricter energy constraints restrict the set of viable routes and alter the structure of the route choice probabilities accordingly.

It is important to highlight that \([1,3,4,5,6,7,2]\) is also the longest path in terms of travel time. This outcome underscores a critical insight: when energy constraints are in place, travelers may be forced to take longer routes in order to ensure sufficient energy availability, particularly through visits to charging stations. In contrast, in the unconstrained setting, this longest path receives the lowest choice probability due to its higher travel cost. This comparison demonstrates that ignoring operational constraints—such as energy limits—can result in significantly biased or incorrect predictions of route choice behavior. Incorporating such constraints into the modeling framework is therefore essential for producing realistic and reliable results.

\subsection{Rechargeable vehicles in a grid network}

\begin{figure}[htb] 
    \centering
    \includegraphics[width=0.6\linewidth]{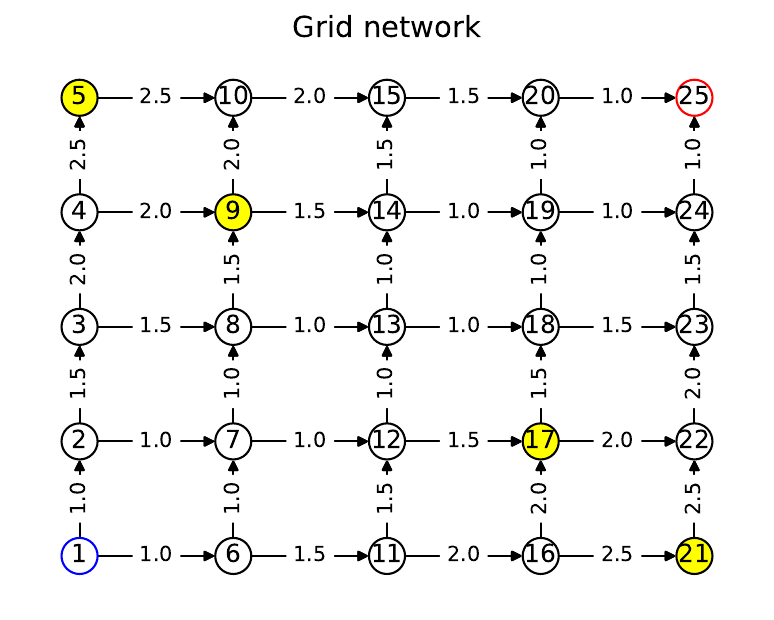}
    \caption{Grid graph with 25 nodes including 4 charging stations.} 
    \label{fig:grid} 
\end{figure}

We now present a larger illustrative example to highlight the impact of feasibility 
constraints in route choice modeling for rechargeable vehicles. The network, shown 
in Figure \ref{fig:grid}, is a $5 \times 5$ grid consisting of $25$ nodes, among which 
$4$ are designated as charging stations. A vehicle travels from node $1$ (bottom left) 
to node $25$ (top right). Each edge is associated with a travel time, and the 
instantaneous utility of traversing an edge is set to the negative of its travel time.  

For the CRL model, we assume that the energy consumed on each link is equal to its 
travel time. The vehicle has a finite energy capacity of $\alpha = 7$, meaning it 
must recharge at a station before the remaining energy is exhausted. This setting 
naturally enforces feasibility constraints: any route that exceeds the energy 
capacity without visiting a station is considered infeasible and is excluded from 
the consideration choice set.  

We compute link choice probabilities under both the RL and CRL models using the 
specified link utilities. The resulting choice probability distributions are 
visualized by varying the thickness of the links in 
\autoref{fig:grid-unconstrained} and \autoref{fig:grid-constrained}. In the RL model 
with an unrestricted choice set, links closer to the diagonal from node $1$ to 
node $25$ dominate, as they correspond to shorter travel times, whereas links 
farther away from the diagonal have lower probabilities.  

In contrast, the CRL model with the energy upper-bound $\alpha=7$ yields drastically 
different probability patterns. Routes near the diagonal, which appear optimal under 
the RL model, become infeasible due to the absence of charging stations in that 
region and therefore disappear from the choice set. Instead, feasible routes that 
include detours to charging stations, often located near the corners, retain 
non-negligible probability mass. This result demonstrates how the CRL model 
successfully captures feasibility constraints that fundamentally alter the set of 
available alternatives, leading to more realistic modeling of traveler behavior in 
energy-constrained settings.

\begin{figure}[htb]
    \centering
    \begin{subfigure}[t]{0.48\linewidth}
        \centering
        \includegraphics[width=\linewidth]{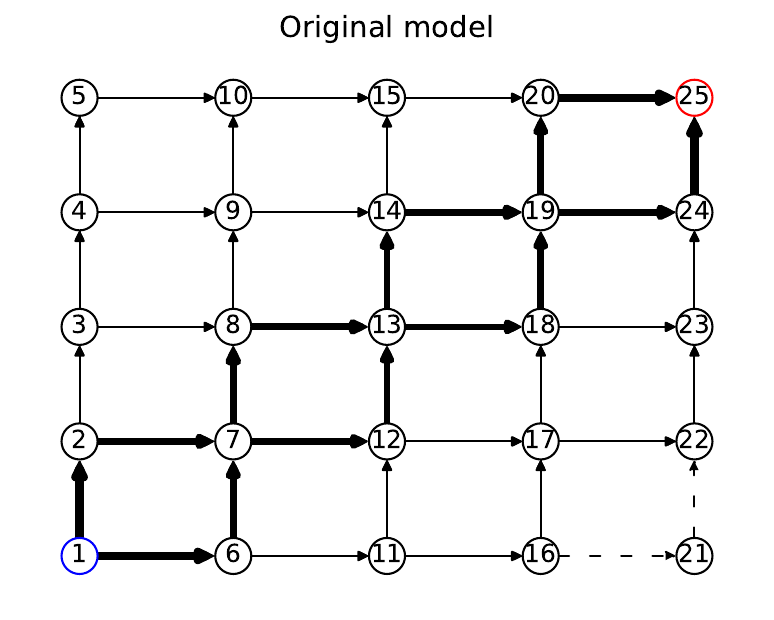}
        \caption{Path probabilities in grid graph with unrestricted choice set.}
        \label{fig:grid-unconstrained}
    \end{subfigure}
    \hfill
    \begin{subfigure}[t]{0.48\linewidth}
        \centering
        \includegraphics[width=\linewidth]{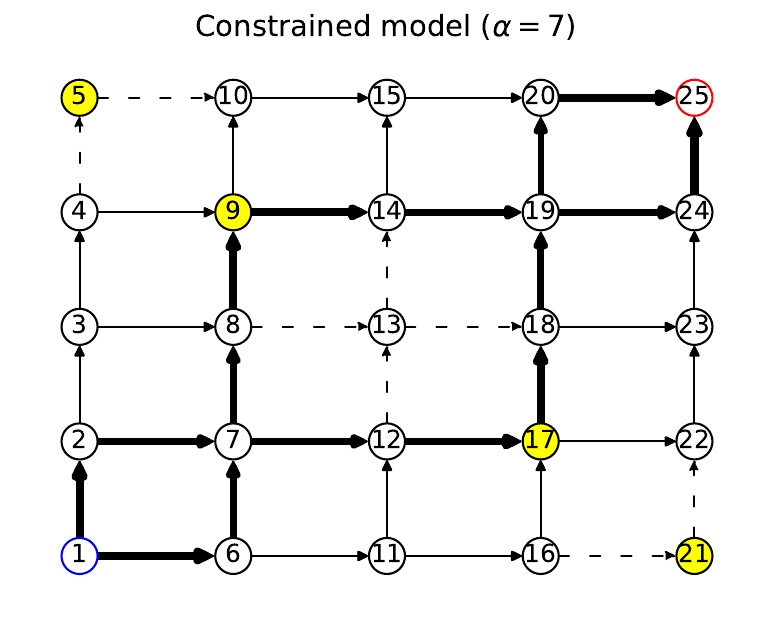}
        \caption{Path probabilities in grid graph with energy constraint of $\alpha=7$.}
        \label{fig:grid-constrained}
    \end{subfigure}
    \caption{Comparison of path probabilities in grid graph under (a) unrestricted choice set and (b) energy-constrained setting.}
    \label{fig:grid-comparison}
\end{figure}

\section{Model Properties}\label{sec:model-properties}
In this section we delve into the main properties of the CRL model described above. 

\subsection{Connection of Path-based MNL Models}
First, to better understand the modeling formulation of the CRL model, let us connect it with a logit-based path choice model. It is well-known that the standard RL model is equivalent to the MNL model over the universal choice set of paths. That is, let $\Omega$ be the set of all possible paths reaching the destination $d$. The path choice probability given by the RL model can be expressed as:
\begin{equation}
    P^{\RL}(\tau) = \frac{e^{\frac{1}{\mu} v(\tau)}}{\sum_{\tau' \in \Omega} e^{\frac{1}{\mu} v(\tau')}},
\end{equation}
where $v(\tau)$ is the accumulated utility of path $\tau$.  Similarly, the proposition below shows that, if the  cost $c(s'|s)$ are non-negative,  the CRL formulation gives path choice probabilities equivalent to an MNL model over a \textit{restricted universal choice set}:
\begin{proposition}\label{prop:equivalence-MNL}
    Assume that $c(s'|s)\geq 0$ for any $s,s'\in 
    \cS$. Let $\Omega^\alpha$ be a restricted universal choice set that contains only paths whose cost does not exceed $\alpha$, i.e., $\Omega^\alpha = \{\tau \in \Omega ~|~ C(\tau) \leq \alpha\}$. Then, the CRL model is equivalent to the MNL model over $\Omega^\alpha$:
    \begin{equation}
        P^{\CRL}(\tau) =\begin{cases}
            \frac{e^{\frac{1}{\mu} v(\tau)}}{\sum_{\tau' \in \Omega^\alpha} e^{\frac{1}{\mu} v(\tau')}} &\text{ if }\tau \in\Omega^\alpha\\
            0&\text{ otherwise }
        \end{cases} 
    \end{equation}
    where $v(\tau)$ is the accumulated utility along path $\tau$, i.e., for any path $\tau = \{s_0,s_1,\ldots,s_T\}$, $v(\tau) = \sum_{t=0}^{T-1} v(s_{t+1}|s_t)$.
\end{proposition}
\begin{proof}
   For any path \( \tau = \{s_0, s_1, \ldots, s_T\} \), let \( S_t = \{s_0, \ldots, s_t\} \) denote the sequence of states from the origin \( s_0 \) up to state \( s_t \). The choice probability of path \( \tau \) under the CRL model can be computed as:
\begin{align}
P(\tau) &= \prod_{t=0}^{T-1} \pi(s_{t+1} \mid S_t) \nonumber \\
&= \prod_{t=0}^{T-1} \frac{\exp\left( \frac{1}{\mu}(v(s_{t+1} \mid s_t) + V(S_{t+1})) \right)}{\sum_{s \in N(s_t)} \exp\left( \frac{1}{\mu}(v(s \mid s_t) + V(S_t \cup \{s\})) \right)}\nonumber \\
&= \prod_{t=0}^{T-1} \frac{\exp\left( \frac{1}{\mu}(v(s_{t+1} \mid s_t) + V(S_{t+1})) \right)}{\exp\left( \frac{1}{\mu} V(S_t) \right)} \nonumber\\
&= \frac{\exp\left( \frac{1}{\mu} v(\tau) \right)}{\exp\left( \frac{1}{\mu} V(S_0) \right)},\label{eq:prob-V(S)}
\end{align}
where \( v(\tau) = \sum_{t=0}^{T-1} v(s_{t+1} \mid s_t) \) is the total utility along path \( \tau \), and \( V(S_0) \) is the value function at the origin.

The value function \( V(S_t) \) can be computed recursively using the following relation:
\begin{equation}\label{eq:recur-V(S)}
\exp\left( \frac{1}{\mu} V(S_t) \right) = 
\begin{cases}
1, & \text{if } s_t = d \text{ and } C(S_t) \leq \alpha, \\
\sum\limits_{s_{t+1} \in N(s_t)} e^ {\frac{1}{\mu} v(s_{t+1} \mid s_t) } \cdot e^{ \frac{1}{\mu} V(S_t \cup \{s_{t+1}\})} , & \text{if } s_t \neq d \text{ and } C(S_t) \leq \alpha, \\
0, & \text{otherwise},
\end{cases}
\end{equation}
for all \( t = 0, 1, \ldots \), where \( C(S_t) \) is the accumulated cost (e.g., energy consumed) along the sequence \( S_t \).

Moreover, assuming the cost function \( c(s' \mid s) \) is non-negative, it follows that for any two sequences \( S_h \subseteq S_k \), if \( C(S_h) > \alpha \), then \( C(S_k) > \alpha \) as well. Consequently, if \( V(S_h) = -\infty \), then \( V(S_k) = -\infty \). As a result, the recursion simplifies to:
\begin{equation}\label{eq:V(S)-recur}
    \exp\left( \frac{1}{\mu} V(S_t) \right) = \sum_{s_{t+1} \in N(s_t)} \exp\left( \frac{1}{\mu} v(s_{t+1} \mid s_t) \right) \cdot \exp\left( \frac{1}{\mu} V(S_t \cup \{s_{t+1}\}) \right),
\end{equation}
for all \( t \) such that \( s_t \neq d \) and \( C(S_t) \leq \alpha \).

This implies:
\[
\exp\left( \frac{1}{\mu} V(S_0) \right) = \sum_{\tau \in \Omega} \exp\left( \frac{1}{\mu} v(\tau) \right) \cdot \exp\left( \frac{1}{\mu} V(\tau) \right),
\]
where \( \Omega \) denotes the set of all possible paths from \( s_0 \) to the destination \( d \), and \( V(\tau) \) is the value associated with the terminal state sequence of path \( \tau \). Notably, \( \exp\left( \frac{1}{\mu} V(\tau) \right) = 0 \) if \( C(\tau) > \alpha \), and \( \exp\left( \frac{1}{\mu} V(\tau) \right) = 1 \) if \( C(\tau) \leq \alpha \). This leads to the final expression:
\[
\exp\left( \frac{1}{\mu} V(S_0) \right) = \sum_{\tau \in \Omega^\alpha} \exp\left( \frac{1}{\mu} v(\tau) \right),
\]
where \( \Omega^\alpha \subseteq \Omega \) is the set of all feasible paths satisfying the  constraint \( C(\tau) \leq \alpha \), as desired.
\end{proof}

We note that Proposition~\ref{prop:equivalence-MNL} no longer holds if the cost function is allowed to take both positive and negative values. The reason is that when \( c(s' \mid s) \) can be either positive or negative, the condition \( C(\tau) \leq \alpha \) is not sufficient to guarantee the feasibility of a path \( \tau \). In particular, even if the total accumulated cost along a path satisfies the constraint, the path may still be infeasible if, at any intermediate state \( s_t \), the accumulated cost \( C(S_t) \) exceeds \( \alpha \). This violates the per-step feasibility condition imposed by the constraint.

To extend Proposition~\ref{prop:equivalence-MNL} to the case where the cost function can take both positive and negative values, we refine the definition of the restricted choice set as follows:

\begin{definition}[Feasible Path Set under Stepwise Constraint]
Let \( \overline{\Omega}^\alpha \) denote the set of all paths from the origin to the destination such that, for every path \( \tau \in \overline{\Omega}^\alpha \), the accumulated cost at every intermediate step does not exceed the threshold \( \alpha \). Formally, we define:
\[
\overline{\Omega}^\alpha = \left\{ \tau = \{s_0, \ldots, s_T\} \, \middle| \, s_T = d,~ C(S_t) \leq \alpha \text{ for all } S_t = \{s_0, \ldots, s_t\},~ t = 0, \ldots, T \right\}.
\]
\end{definition}
Under this definition, we can show that the CRL model is also equivalent to a MNL model over the restricted choice set $\overline{\Omega}^\alpha$.
\begin{proposition}\label{prop:equi-MNL-2}
The path probabilities induced by the CRL model are equivalent to those given by a MNL model defined over the restricted feasible set \( \overline{\Omega}^\alpha \):
\begin{equation}
P^{\text{CRL}}(\tau) = 
\begin{cases}
\frac{\exp\left( \frac{1}{\mu} v(\tau) \right)}{\sum\limits_{\tau' \in \overline{\Omega}^\alpha} \exp\left( \frac{1}{\mu} v(\tau') \right)} & \text{if } \tau \in \overline{\Omega}^\alpha, \\
0 & \text{otherwise}.
\end{cases}
\end{equation}   
\end{proposition}

\begin{proof}
We begin by leveraging the recursive formulation from Equation~\ref{eq:V(S)-recur}, which remains valid even when the cost function may take negative values:
\[
\exp\left( \frac{1}{\mu} V(S_t) \right) = \sum_{s_{t+1} \in N(s_t)} \exp\left( \frac{1}{\mu} v(s_{t+1} \mid s_t) \right) \cdot \exp\left( \frac{1}{\mu} V(S_t \cup \{s_{t+1}\}) \right),
\]
where \( V(S_t) \) denotes the value function at state \( S_t \), and \( N(s_t) \) denotes the set of next feasible nodes from \( s_t \).

Importantly, whenever the accumulated cost \( c(S_t) > \alpha \), we have:
$\exp\left( \frac{1}{\mu} V(S_t) \right) = 0.$
Therefore, the total value function at the origin state \( S_0 \) becomes:
\[
\exp\left( \frac{1}{\mu} V(S_0) \right) = \sum_{\tau \in \Omega} \exp\left( \frac{1}{\mu} v(\tau) \right) \cdot \Psi(\tau),
\]
where the weight function \( \Psi(\tau) \) is computed as:
\[
\Psi(\tau) =
\begin{cases}
0 & \text{if there exists a subsequence } S' \subseteq \tau \text{ such that } c(S') > \alpha, \\
\exp\left( \frac{1}{\mu} V(\tau) \right) & \text{otherwise}.
\end{cases}
\]
As a result, only paths in the restricted feasible set \( \overline{\Omega}^\alpha \) contribute to the sum, yielding:
\[
\exp\left( \frac{1}{\mu} V(S_0) \right) = \sum_{\tau \in \overline{\Omega}^\alpha} \exp\left( \frac{1}{\mu} v(\tau) \right).
\]
Finally, combining this with the probabilistic formulation in Equation~\ref{eq:prob-V(S)}, we obtain the desired result:
\[
P^{\text{CRL}}(\tau) = \frac{\exp\left( \frac{1}{\mu} v(\tau) \right)}{\sum\limits_{\tau' \in \overline{\Omega}^\alpha} \exp\left( \frac{1}{\mu} v(\tau') \right)}.
\]
\end{proof}

A direct implication of Proposition~\ref{prop:equi-MNL-2} is that, for any given path \( \tau \in \overline{\Omega}^\alpha \), the path choice probability assigned by the original RL model is less than or equal to that assigned by the CRL model. As a result, the in-sample log-likelihood under the CRL model is always greater than or equal to that under the standard RL model. We formally state this result in the following corollary.

\begin{corollary}
For any feasible path \( \tau \in \overline{\Omega}^\alpha \), we have
$P^{\text{RL}}(\tau) \leq P^{\text{CRL}}(\tau).$
Consequently, given a dataset of observed trajectories \( \mathcal{D} = \{\tau_1, \ldots, \tau_n\} \), if all observed paths satisfy the constraint (i.e., \( \tau_i \in \overline{\Omega}^\alpha \) for all \( i \)), then the CRL model yields a log-likelihood value that is greater than or equal to that of the RL model:
\[
\mathcal{L}^{\text{CRL}}(\mathcal{D}) \geq \mathcal{L}^{\text{RL}}(\mathcal{D}),
\]
where the log-likelihood functions are defined as
\[
\mathcal{L}^{\text{RL}}(\mathcal{D}) = \sum_{\tau \in \mathcal{D}} \ln P^{\text{RL}}(\tau), \quad 
\mathcal{L}^{\text{CRL}}(\mathcal{D}) = \sum_{\tau \in \mathcal{D}} \ln P^{\text{CRL}}(\tau).
\]
\end{corollary}
The corollary follows directly from the fact that the restricted choice set \( \overline{\Omega}^\alpha \) is a subset of the universal choice set \( \Omega \), which immediately implies:
\[
\sum_{\tau \in \overline{\Omega}^\alpha} \exp\left( \frac{1}{\mu} v(\tau) \right) \leq \sum_{\tau \in \Omega} \exp\left( \frac{1}{\mu} v(\tau) \right).
\]
This inequality directly leads to the result \( P^{\text{RL}}(\tau) \leq P^{\text{CRL}}(\tau) \) for any feasible path \( \tau \in \overline{\Omega}^\alpha \).  Equality holds, i.e., \( P^{\text{RL}}(\tau) = P^{\text{CRL}}(\tau) \) for all \( \tau \), if and only if the restricted and universal choice sets are identical, that is, \( \overline{\Omega}^\alpha = \Omega \). In this case, all paths are feasible with respect to the constraint, and the CRL model reduces to the standard RL model. In fact, we can see that the inequality above is strict if the universal choice set 
$\Omega$ strictly contains the feasible set $\overline{\Omega}^{\alpha}$, i.e.,
$\mathcal{L}^{\text{CRL}}(\mathcal{D}) > \mathcal{L}^{\text{RL}}(\mathcal{D}),$
whenever there exists a path $\tau \in \Omega$ that is not feasible with respect to the constraint 
($\tau \notin \overline{\Omega}^{\alpha}$).

\section{Markovian Choice Probabilities through Extended MDP}\label{sec:extended-state-space}

As discussed above, the CRL model is inherently non-Markovian, since the choice 
probabilities at a given state (i.e., a node in the network) depend not only on the 
current state but also on the sequence of historical states visited. This violation of 
the Markov property renders the model impractical for estimation and prediction.
To address this issue, in the following section we propose an approach based on 
\emph{state-space expansion}. The key idea is to augment the state representation 
with the information regarding the accumulated cost, thereby transforming the CRL model 
into an equivalent Markovian model. Once the Markov property is restored, the 
problem can be reformulated as a standard RL model, enabling 
efficient estimation through classical methods such as value iteration.
We begin by presenting our approach to transform the model into a Markovian form, 
followed by a discussion of the associated estimation methods.

\subsection{Markovian Choice Probabilities}
To achieve Markovian choice probabilities, we leverage the observation that the 
expected utility at a state $s_t$ depends only on the accumulated cost at that state, 
rather than on the entire history of previously visited states. This motivates us to 
incorporate the accumulated cost into the state representation, thereby ensuring 
that the value function is independent of the historical trajectory. 

Formally, we define an \emph{extended state space} $\widetilde{\mathcal{S}}$, where 
each extended state $\widetilde{s} \in \widetilde{\mathcal{S}}$ is a pair 
$\widetilde{s} = (s,z),$
with $s \in \mathcal{S}$ denoting the original state and $z \in \mathbb{R}$ the 
accumulated cost up to that state. Similarly, we define the set of feasible next 
states in the extended space as
\[
\widetilde{N}(\widetilde{s}) 
= \bigl\{ (s',z') \;\big|\; s' \in N(s), \; z' = z + c(s' \mid s) \bigr\},
\]
that is, each successor in the extended state space consists of a next state $s'$ 
from the original state space together with the updated accumulated cost $z' = z+c(s'|s)$. 

We now formulate a RL model on this extended state space. 
At each extended state $\widetilde{s}_t = (s_t,z_t)$, the decision-maker selects 
the next state $\widetilde{s}_{t+1} = (s_{t+1},z_{t+1})$ by maximizing the expected utility
\[
\mathbb{E}_{\epsilon} \left[ \max_{\widetilde{s}_{t+1} \in \widetilde{N}(\widetilde{s}_t)} 
\Bigl\{ \widetilde{v}(\widetilde{s}_{t+1} \mid \widetilde{s}_t) 
+ \widetilde{V}(\widetilde{s}_{t+1}) 
+ \mu \, \epsilon(\widetilde{s}_{t+1}) \Bigr\} \right],
\]
where $\widetilde{v}(\widetilde{s}_{t+1} \mid \widetilde{s}_t)$ denotes the deterministic 
immediate utility of moving from $\widetilde{s}_t$ to $\widetilde{s}_{t+1}$, given by
\[
\widetilde{v}(\widetilde{s}_{t+1} \mid \widetilde{s}_t) = v(s_{t+1} \mid s_t),
\]
$\widetilde{V}(\widetilde{s}_{t+1})$ is the expected utility value from the extended state 
$\widetilde{s}_{t+1}$ to the destination, and $\epsilon(\widetilde{s}_{t+1})$ are 
i.i.d.\ Gumbel-distributed (Type I extreme value) random components.  

The value function $\widetilde{V}(\widetilde{s}_t)$ associated with state  $\widetilde{s}_t = (s_t,z_t)$ is defined recursively as
\begin{equation}\label{eq:extended-V-recur}
  \widetilde{V}(\widetilde{s}_t) =
\begin{cases}
    0, & \text{if } s_t = d \text{ and } z_t \leq \alpha, \\[6pt]
    -\infty, & \text{if } z_t > \alpha, \\[6pt]
    \mu \ln \left( 
        \sum_{\widetilde{s}_{t+1} \in \widetilde{N}(\widetilde{s}_t)} 
        \exp\!\left( \tfrac{1}{\mu} \bigl[ 
            \widetilde{v}(\widetilde{s}_{t+1} \mid \widetilde{s}_t) 
            + \widetilde{V}(\widetilde{s}_{t+1}) 
        \bigr] \right) 
    \right), & \text{otherwise}.
\end{cases}  
\end{equation}
In other words, the value function is assigned $-\infty$ whenever the accumulated 
cost at the current state exceeds the upper bound $\alpha$. In all other cases, it 
takes the standard log-sum-exp form of the RL 
model. This construction ensures that any path whose cumulative cost violates 
the constraint is rendered infeasible and, consequently, is assigned zero probability 
of being chosen in the decision process. In effect, the extended state-space 
formulation enforces the feasibility constraints directly within the recursive 
structure of the value function, thereby integrating the constraint into the 
Markovian dynamics of the model.

The choice probability at each state is given by the standard MNL model:
\begin{equation}
    \widetilde{\pi} \bigl(\widetilde{s}_{t+1}\mid \widetilde{s}_t\bigr) 
    = 
    \begin{cases}
        \dfrac{\exp\!\left( \tfrac{1}{\mu} \left[ 
            \widetilde{v}(\widetilde{s}_{t+1} \mid \widetilde{s}_t) 
            + \widetilde{V}(\widetilde{s}_{t+1})
        \right] \right)}
        {\sum\limits_{\widetilde{s} \in \widetilde{N}(\widetilde{s}_t)} 
        \exp\!\left( \tfrac{1}{\mu} \left[ 
            \widetilde{v}(\widetilde{s} \mid \widetilde{s}_t) 
            + \widetilde{V}(\widetilde{s})
        \right] \right)}, 
        & \text{if } \widetilde{s}_{t+1} \in \widetilde{N}(\widetilde{s}_t), \\[12pt]
        0, & \text{otherwise}.
    \end{cases}
\end{equation}
For notational convenience, we denote the recursive logit model defined on the extended 
state space by \textbf{ERL}. We are now ready to state the main result, which establishes 
the equivalence between the ERL and the CRL model.

\begin{theorem}\label{th:ERL-equipvalence}
The ERL model is equivalent to the CRL model described in Section~\ref{sec:CRL}, in the sense that for 
any path $\tau$ in the original state space $\mathcal{S}$, there exists a one-to-one mapping 
to a path $\widetilde{\tau}$ in the extended state space such that their choice probabilities coincide:
\[
P^{\CRL}(\tau) \;=\; P^{\ERL}(\widetilde{\tau}).
\]
\end{theorem}

\begin{proof}
The proof proceeds by showing that both probability formulations can be expressed as 
path-based probabilities over a restricted universal choice set.  As established in Proposition~\ref{prop:equi-MNL-2}, the CRL model assigns the following probability to a path $\tau$:
\begin{equation}
P^{\CRL}(\tau) = 
\begin{cases}
\dfrac{\exp\!\left( \tfrac{1}{\mu} v(\tau) \right)}
{\sum\limits_{\tau' \in \overline{\Omega}^\alpha} 
\exp\!\left( \tfrac{1}{\mu} v(\tau') \right)}, & \text{if } \tau \in \overline{\Omega}^\alpha, \\[12pt]
0, & \text{otherwise},
\end{cases}
\end{equation}
where $v(\tau)$ denotes the total utility along path $\tau$, $\overline{\Omega}^\alpha$ is the restricted 
universal choice set under the cost constraint.  

Now consider any path $\tau = [s_0,\ldots,s_T]$ in the original state space. This path uniquely 
corresponds to a path in the extended state space,
\[
\widetilde{\tau} = [\widetilde{s}_0,\ldots,\widetilde{s}_T] 
\;\stackrel{\text{def}}{=}\; 
[(s_0,z_0),\ldots,(s_T,z_T)],
\]
where $z_0 = 0$ and, for $t=1,\ldots,T$, the accumulated cost is defined recursively as
\[
z_t = \sum_{k=0}^{t-1} c(s_{k+1}\mid s_k).
\]

For notational clarity, let $\mathcal{C}(\widetilde{s})$ denote the accumulated cost associated 
with an extended state $\widetilde{s}$, i.e., $\mathcal{C}(\widetilde{s}_t) = z_t$.  
The probability of $\widetilde{\tau}$ under the ERL model is given by the recursive logit formulation:
\[
P^{\ERL}(\widetilde{\tau}) = 
\frac{\exp\!\left(\tfrac{1}{\mu}\sum_{t=0}^{T-1}
\widetilde{v}(\widetilde{s}_{t+1}\mid \widetilde{s}_t)\right)}
{\exp\!\left(\tfrac{1}{\mu}\widetilde{V}(\widetilde{s}_0)\right)}.
\]

From the recursive definition of the extended value function, we have
\[
\exp\!\left(\tfrac{1}{\mu}\widetilde{V}(\widetilde{s}_t)\right) 
= \sum_{\widetilde{s}_{t+1}\in \widetilde{N}(\widetilde{s}_t)} 
\exp\!\left( \tfrac{1}{\mu}\widetilde{v}(\widetilde{s}_{t+1}\mid \widetilde{s}_t)\right) 
\exp\!\left(\tfrac{1}{\mu}\widetilde{V}(\widetilde{s}_{t+1})\right),
\]
whenever $\mathcal{C}(\widetilde{s}_t)\leq \alpha$, and 
$\exp(\tfrac{1}{\mu}\widetilde{V}(\widetilde{s}_t))=0$ otherwise.  
Iterating this recursion yields
\[
\exp\!\left(\tfrac{1}{\mu}\widetilde{V}(\widetilde{s}_0)\right) 
= \sum_{\widetilde{\tau}\in \widetilde{\Omega}} 
\exp\!\left(\tfrac{1}{\mu}\widetilde{v}(\widetilde{\tau})\right)\Psi(\widetilde{\tau}),
\]
where $\widetilde{v}(\widetilde{\tau}) = \sum_{(s,s')\in \tau}\widetilde{v}(s'\mid s)$ and 
$\Psi(\widetilde{\tau})$ is an indicator function, equal to $0$ if there exists any 
$\widetilde{s}\in \widetilde{\tau}$ with $\mathcal{C}(\widetilde{s}) > \alpha$, and $1$ otherwise.  
Accordingly, we may equivalently restrict the summation to the feasible path set:
\[
\widetilde{\Omega}^\alpha = \left\{ \widetilde{\tau}\in \widetilde{\Omega} \;\middle|\; 
\mathcal{C}(\widetilde{s}) \leq \alpha, \;\; \forall \widetilde{s}\in \widetilde{\tau} \right\},
\]
so that
\[
\exp\!\left(\tfrac{1}{\mu}\widetilde{V}(\widetilde{s}_0)\right) 
= \sum_{\widetilde{\tau}\in \widetilde{\Omega}^\alpha} 
\exp\!\left(\tfrac{1}{\mu}\widetilde{v}(\widetilde{\tau})\right).
\]
Therefore, the ERL path probability takes the form
\[
P^{\ERL}(\widetilde{\tau}) 
= \frac{\exp\!\left(\tfrac{1}{\mu}\sum_{t=0}^{T-1}\widetilde{v}(\widetilde{s}_{t+1}\mid \widetilde{s}_t)\right)}
{\sum\limits_{\widetilde{\tau}\in \widetilde{\Omega}^\alpha} 
\exp\!\left(\tfrac{1}{\mu}\widetilde{v}(\widetilde{\tau})\right)}.
\]
Finally, note that for the corresponding original path $\tau = [s_0,\ldots,s_T]$, 
we have
\[
\sum_{t=0}^{T-1}\widetilde{v}(\widetilde{s}_{t+1}\mid \widetilde{s}_t)
= \sum_{t=0}^{T-1} v(s_{t+1}\mid s_t).
\]
Thus, the two path probabilities coincide:
\[
P^{\ERL}(\widetilde{\tau}) \;=\; P^{\CRL}(\tau),
\]
establishing the desired equivalence.
\end{proof}

Theorem~\ref{th:ERL-equipvalence} establishes the equivalence between the 
non-Markovian CRL model and the Markovian ERL model. 
This equivalence is crucial, as it implies that estimation can be carried out on the 
extended state space using standard recursive logit estimation techniques, such as 
the NFXP algorithm \citep{Rust87}.  In the following, we provide a detailed discussion of the estimation procedure 
for the ERL model. In particular, we highlight an important property: under certain 
conditions, estimation based on the ERL formulation exhibits greater numerical 
stability compared to estimation of an unconstrained RL model.  
This stability arises because the ERL representation naturally enforces feasibility 
constraints through the extended state space, thereby eliminating infeasible paths 
and reducing irregularities in the likelihood function.

\subsection{Estimation Methods}
Since the ERL model takes the standard form of the RL model, 
its estimation can be carried out using the NFXP algorithm. 
In this approach, value iteration is employed to compute both the value function 
$\widetilde{V}(\widetilde{s})$ and the corresponding likelihood function.  
 To make estimation computationally feasible, we impose an additional assumption 
on the cost structure. Specifically, we assume that the cost $c(s' \mid s)$ takes 
values in a discrete set. This assumption guarantees that the extended state space 
$\widetilde{\mathcal{S}}$ is finite and discrete, thereby enabling practical 
implementation of the NFXP algorithm on the ERL model.

The recursive definition of $\widetilde{V}(\widetilde{s})$ in \eqref{eq:extended-V-recur} allows us to define 
\[
\widetilde{Z}(\widetilde{s}) \;=\; \exp\!\left(\tfrac{1}{\mu}\widetilde{V}(\widetilde{s})\right),
\]
and rewrite the Bellman recursion in terms of $\widetilde{Z}(\cdot)$ as
\[
\widetilde{Z}(\widetilde{s}) =
\begin{cases}
    1, & \text{if } \mathcal{N}(\widetilde{s}) = d \text{ and } \mathcal{C}(\widetilde{s}) \leq \alpha, \\[6pt]
    0, & \text{if } \mathcal{C}(\widetilde{s}) > \alpha, \\[6pt] 
        \sum\limits_{\widetilde{s}' \in \widetilde{N}(\widetilde{s})} 
        \exp\!\left( \tfrac{1}{\mu} 
            \widetilde{v}(\widetilde{s}' \mid \widetilde{s}) 
        \right)\cdot \widetilde{Z}(\widetilde{s}'), & \text{otherwise},
\end{cases}
\]
where $\mathcal{N}(\widetilde{s})$ denotes the original state $s \in \mathcal{S}$ corresponding 
to the extended state $\widetilde{s}$.  

Equivalently, we may encode this system in matrix form. First, for any state 
$\widetilde{s}$ such that $\mathcal{C}(\widetilde{s})>\alpha$, we set 
$\widetilde{N}(\widetilde{s})=\emptyset$, effectively removing infeasible states from 
the recursion. Then we  define a vector $B \in \mathbb{R}^{|\widetilde{\mathcal{S}}|}$ with entries 
  \[
  B(\widetilde{s}) = \begin{cases}
  1, & \text{if } \mathcal{N}(\widetilde{s})=d,\\
  0, & \text{otherwise},
  \end{cases}
  \]
and a nonnegative matrix $M \in \mathbb{R}^{|\widetilde{\mathcal{S}}|\times|\widetilde{\mathcal{S}}|}$ with elements
  \[
  M(\widetilde{s},\widetilde{s}') 
  = \begin{cases}
  \exp\!\left(\tfrac{1}{\mu}\widetilde{v}(\widetilde{s}' \mid \widetilde{s})\right), 
  & \text{if } \widetilde{s}'\in \widetilde{N}(\widetilde{s}), \\[6pt]
  0, & \text{otherwise}.
  \end{cases}
  \]
With these definitions, the Bellman recursion can be written compactly as
\[
   \widetilde{Z} \;=\; M \widetilde{Z} + B,
\]
which is a linear system of equations. This form can be solved efficiently via value 
iteration or directly through matrix inversion. Importantly, the system admits a unique 
solution under some certain conditions, as stated in the following proposition.  
\begin{proposition}\label{prop:positive-cost}
Assume that the cost satisfies $c(s,s') > 0$ for all $s,s' \in \mathcal{S}$. 
Then the matrix $(I - M)$ is invertible, where $I$ denotes the identity matrix 
(i.e., a diagonal matrix with ones on the main diagonal and zeros elsewhere). 
Consequently, the linear system
$\widetilde{Z} = M \widetilde{Z} + B$
admits a unique solution, which can be written explicitly as
$\widetilde{Z} = (I - M)^{-1} B.$
\end{proposition}
\begin{proof}
By construction, each entry $M(\widetilde{s},\widetilde{s}')$ is strictly positive only 
if $\widetilde{s}' \in \widetilde{N}(\widetilde{s})$, and in that case it is equal to 
$\exp(\tfrac{1}{\mu}\widetilde{v}(\widetilde{s}' \mid \widetilde{s}))$. Since 
$c(s,s')>0$ for all $(s,s')$, the accumulated cost strictly 
increases along any path in the extended state space. This implies that no cycles can 
exist in the transition structure of $M$, as eventually the accumulated cost exceeds 
$\alpha$, forcing termination.  

Hence, the matrix $M$ is strictly upper block-triangular after a suitable permutation 
of states, and its spectral radius satisfies $\rho(M)<1$. It follows that the Neumann 
series
$(I-M)^{-1} = \sum_{k=0}^\infty M^k$
converges. Therefore, $(I-M)$ is invertible, and the linear system admits the unique 
solution
$\widetilde{Z} = (I-M)^{-1} B.$
\end{proof}

An important implication of Proposition~\ref{prop:positive-cost} is that the estimation of the CRL and ERL 
models can be carried out more conveniently compared to the original recursive logit (RL) model. 
Specifically, prior work has shown that the estimation of the RL model may suffer from instability 
when the utility function takes large values, which can render the associated Bellman equations 
infeasible~\citep{MaiFrejinger22}.  By contrast, Proposition~\ref{prop:positive-cost} establishes that if the cost associated with each transition between 
states is strictly positive---a condition that is generally satisfied in practical applications---then 
the Bellman equations of the ERL model always admit a unique solution, regardless of the scale of 
the utility function $v(s' \mid s)$. This property ensures that the ERL formulation avoids the 
degeneracies that may arise in the standard RL setting and thereby provides a stable foundation 
for consistent estimation.

\section{Numerical experiments}\label{sec:exp}
In this section, we present a series of experiments to evaluate the effectiveness of the 
proposed CRL model in comparison with the standard RL model. The evaluation is carried out 
on both synthetic and real-world datasets. All experiments are conducted on a Google Cloud 
TPU v2-8 server equipped with a 96-core Intel Xeon 2.00~GHz CPU.

\subsection{Analysis on Synthetic Networks and Datasets}

We first conduct experiments on synthetic datasets in order to evaluate the 
performance of the proposed CRL model, compared to the standard RL model. 
Two problem settings of particular interest are considered:  
\begin{enumerate}
    \item \textbf{Route choice with travel time upper bounds.}  
    In this setting, travelers are assumed to only consider routes whose total 
    travel time does not exceed a pre-specified threshold. This reflects 
    practical scenarios where travelers face strict deadlines, such as arriving 
    on time for flights, meetings, or deliveries.  
    \item \textbf{Route choice with rechargeable vehicles.}  
    Here, the traveler must ensure that the remaining battery energy of the 
    vehicle never drops below zero along any chosen route. This captures 
    important features of modern transportation systems where electric or 
    hybrid vehicles are increasingly common, and route feasibility depends not 
    only on travel time but also on energy availability.
\end{enumerate}

\paragraph{Network generation.}  For computational tractability, we construct synthetic networks in the form of 
directed acyclic graphs (DAGs). The DAGs are generated as random geometric graphs 
\citep{Penrose2003RandomGeom}, a topology frequently employed for benchmarking 
graph algorithms. In this setting, node locations are sampled uniformly at random 
within a unit square $[0,1]\times[0,1]$, and edges are created by connecting nodes 
whose Euclidean distance falls below a prescribed threshold.
For each successive link, we associate four attributes---\textit{travel time, 
left turn, right turn, and U-turn}---derived from the geometric positions of the 
incident nodes and their relative orientations. The travel time is proportional to 
the Euclidean length of the link, while turn indicators are determined from the 
angular relationship between consecutive links. This setup provides a controlled, 
yet flexible, framework for simulating route choice scenarios with heterogeneous 
link characteristics while maintaining a mathematically tractable network 
structure.
To generate the observational 
data, we fix a ground-truth parameter vector
$\beta = [-4, \; -0.1, \; -0.05, \; -0.3]^{\top}.$
We then simulate path choices according to the RL formulation.  

\paragraph{Experimental setup.}  The number of observations used for parameter estimation and in-sample prediction 
is set to $3000$, while an additional $1000$ observations are generated for 
out-of-sample prediction. Parameter estimation is conducted by maximizing the 
likelihood of the in-sample dataset, yielding estimated parameters $\beta^*$.  

\paragraph{Evaluation metrics.}  The quality of estimation is evaluated using log-likelihood measures.  
Given an observation dataset 
\(\mathcal{D} = \{\tau_1, \ldots, \tau_n\}\), we define the average log-likelihood 
under the estimated parameter vector $\beta^*$ as $\mathcal{L}^{RL}(\mathcal{D}\mid \beta^*) = \frac{1}{n} \sum_{i=1}^n\ln P^{\RL}(\tau_i)$. 
Since log-likelihood values are always non-positive, values closer to zero 
indicate better alignment of the model with the observed data.  
Furthermore, to quantify the relative performance gain of CRL over RL, we further define the 
percentage improvement metric as:
\begin{equation}\nonumber
    \%Improve(\mathcal{D}) = 
  \frac{  \cL^{CRL}(\mathcal{D}) -  \cL^{RL}(\mathcal{D})}{|\cL^{RL}(\mathcal{D})|} 
    \times 100\%.
\end{equation}
A positive value of $\%Improve$ indicates that the CRL model achieves better performance (in terms of likelihood prediction) compared to the standard RL model.  

The synthetic experiments enable us to systematically assess the effect of 
incorporating feasibility constraints (such as travel time thresholds or energy 
limits) into the choice model. In the standard RL formulation, unreasonable paths 
--- for example, those that violate energy constraints or whose travel times far 
exceed practical limits --- may still receive positive choice probabilities. This 
undesirable feature often results in biased parameter estimates and degraded 
predictive accuracy. By contrast, the CRL model explicitly removes infeasible paths 
from the universal choice set, ensuring that only feasible routes are considered in 
the estimation process. This leads to more robust parameter recovery and improved 
generalization capability. Consequently, the CRL model is expected to  
outperform the RL model when constraints are taken into consideration, in both for in-sample estimation and out-of-sample prediction

\subsubsection{Route choice with travel time upper-bounds}
In this experiment, we consider a route choice setting in which the traveler imposes 
an upper bound on the total travel time. The networks and associated observations 
are generated as follows.  

\paragraph{Graph generation.}  
We construct directed acyclic graphs (DAGs) of varying sizes, specifically with 
$20$, $30$, $40$, and $50$ nodes. For each size, $5$ independent DAGs are generated 
using the random geometric graph method. More precisely, $N$ nodes are uniformly 
distributed within a $1 \times 1$ geometric space. An edge is drawn between two 
nodes if their Euclidean distance is smaller than $\frac{2}{\sqrt{N}}$. Additional 
edges are added if necessary to ensure that the graph is connected. The node with 
the smallest index is designated as the source, the node with the largest index as 
the destination, and all edges are oriented from smaller to larger indices, thereby 
ensuring acyclicity.  

\paragraph{Observation generation.}  
Let $T_{\max}$ denote the longest feasible travel time of any route from source to 
destination in a given DAG. The cost associated with each link is interpreted as the 
travel time on that edge. To impose feasibility, we introduce an upper bound 
$\alpha$ on the total travel time that a traveler is willing to accept, defined as  
$\alpha = \%threshold \times T_{\max},$
where $\%threshold$ is a parameter specifying the ratio of the upper bound relative 
to the longest possible route. We consider eight levels of $\%threshold$ ranging 
from $20\%$ to $90\%$.  

Once $\alpha$ is specified, we generate route choice observations under two models:  
\begin{itemize}
    \item \textbf{Recursive Logit (RL).} Observations are generated using the standard RL model. 
    For consistency with the imposed travel time constraint, we discard any path 
    samples whose total travel time exceeds $\alpha$, as such routes are deemed 
    unrealistic for the traveler.  
    \item \textbf{Constrained Recursive Logit (CRL).} Observations are generated using the CRL 
    formulation, where infeasible paths (i.e., with travel time greater than $\alpha$) 
    are systematically excluded from the universal choice set.  
\end{itemize}
\paragraph{Experimental protocol.}  
For each DAG and each $\%threshold$ setting, we conduct $10$ independent trials. 
In each trial, $3000$ in-sample and $1000$ out-of-sample route observations are 
generated. The estimated parameters obtained from the in-sample observations are 
then used to evaluate both in-sample and out-of-sample log-likelihood performance.  

To quantify the improvement of CRL over RL, we compute the percentage improvement 
metric $\%Improve$ (defined in the previous section). For each configuration, the 
average $\%Improve$ across $10$ trials is reported. This provides a robust 
comparison of the two models across varying network sizes and travel time 
thresholds, highlighting the conditions under which the CRL model yields the 
largest performance gains.



\paragraph{Comparison results.}  
The estimation results obtained when the RL model is used to generate observations 
are shown in \autoref{fig:u_unconstrained}. Out-of-sample results are almost 
indistinguishable from in-sample ones, indicating that both the observation size 
and the number of trials are sufficiently large to ensure consistent parameter 
estimation.  

\begin{figure}[htb] 
    \centering
    \begin{subfigure}{0.4\textwidth}
        \includegraphics[width=\linewidth]{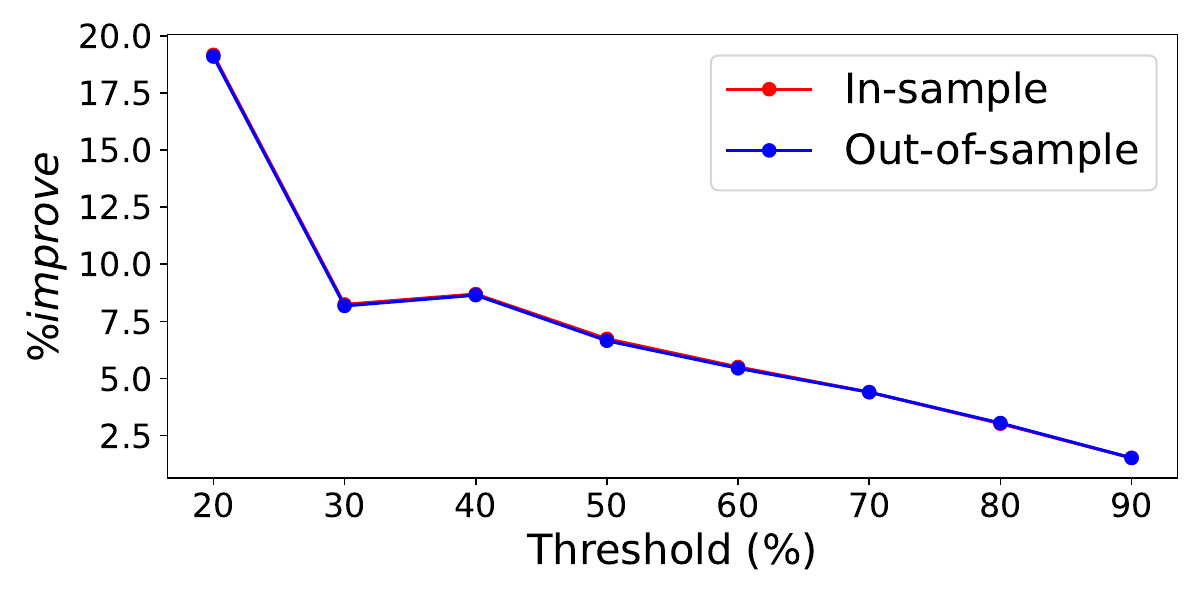}
        \caption{20 nodes}
    \end{subfigure}
    \begin{subfigure}{0.4\textwidth}
        \includegraphics[width=\linewidth]{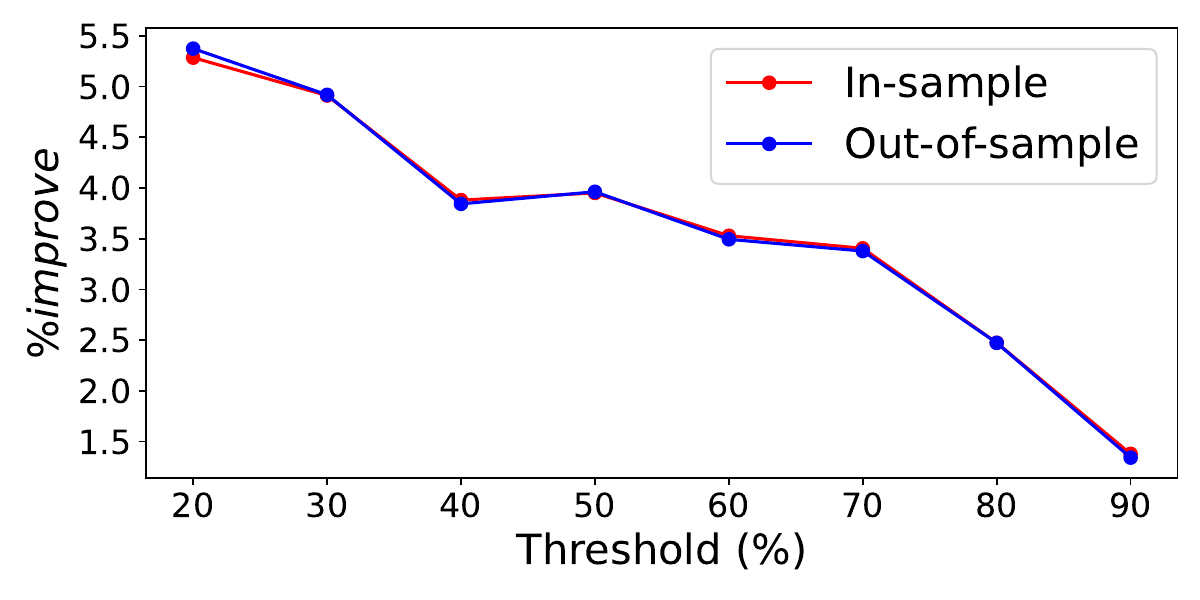}
        \caption{30 nodes}
    \end{subfigure}
    \begin{subfigure}{0.4\textwidth}
        \includegraphics[width=\linewidth]{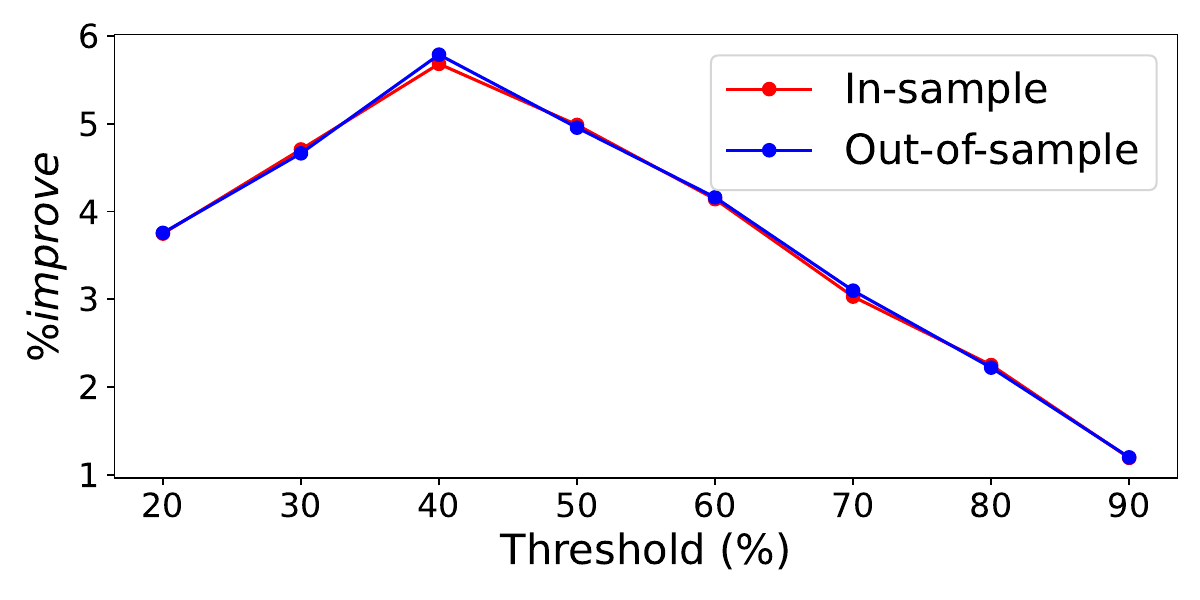}
        \caption{40 nodes}
    \end{subfigure}
    \begin{subfigure}{0.4\textwidth}
        \includegraphics[width=\linewidth]{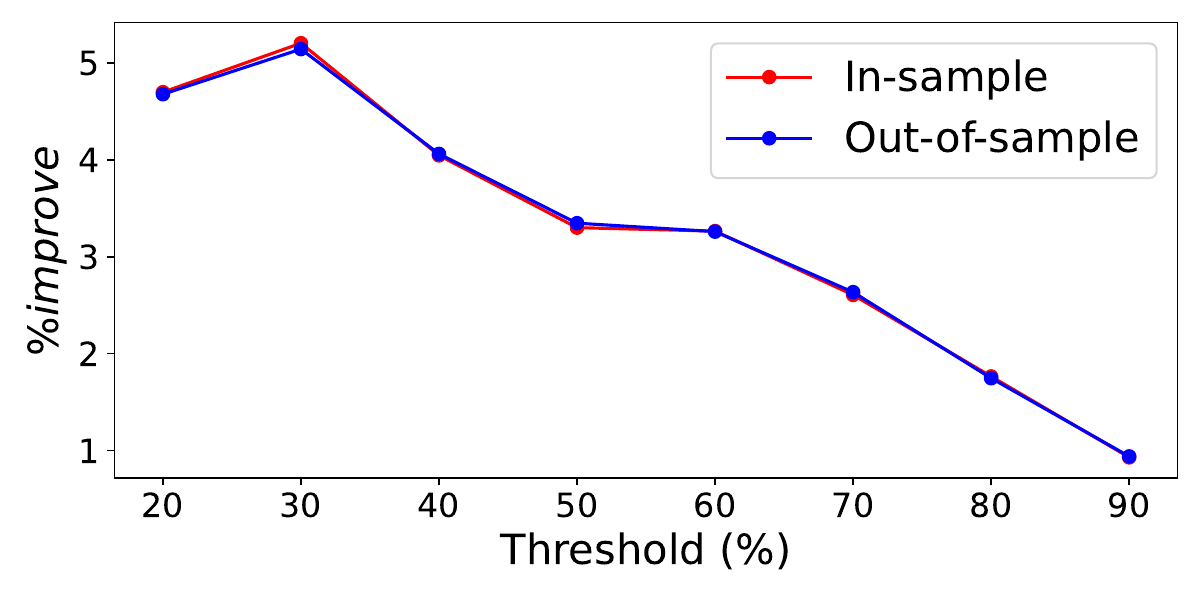}
        \caption{50 nodes}
    \end{subfigure}
    \caption{Estimation on synthetic observations generated by the RL model.} 
    \label{fig:u_unconstrained} 
\end{figure}

In terms of performance, the CRL model consistently achieves better (i.e., higher) 
log-likelihood values than the standard RL model. However, the performance gap 
between the two models decreases as the upper-bound threshold increases. This 
behavior is expected: when the threshold approaches $100\%$, nearly all routes 
remain feasible, and the CRL formulation reduces to the unconstrained RL model. 
Hence, under such conditions, the two models become virtually identical.  At the other extreme, when the threshold is very tight (e.g., $20\%$ or $30\%$ of 
$T_{\max}$), the improvement of CRL over RL is less pronounced and even fluctuates 
across trials. This effect can be attributed to the limited number of feasible 
routes available under such strict constraints, which reduces the variability of 
the choice set and may cause instability in estimation.  


%

When the observations are generated by the CRL model, the percentage improvements 
are reported in \autoref{fig:u_constrained}. In this case, the downward trend of 
$\%Improve$ is more stable: the improvement decreases smoothly from approximately 
$8$--$10\%$ at $\%threshold=20\%$ to nearly $0\%$ at $\%threshold=90\%$. This 
confirms the expected behavior that, as the upper bound becomes less restrictive, 
the CRL and RL models converge to similar performance.  

Another notable observation is that the performance gap between CRL and RL becomes 
less pronounced as the size of the DAG increases. This phenomenon is natural, since 
larger graphs imply exponentially larger route choice sets. With the number of 
observations fixed at $3000$, the sampled trajectories cover only a small fraction 
of the feasible choice set in larger graphs. As a result, the available data are 
less informative, and the distinction in performance between CRL and RL diminishes.

\begin{figure}[htb] 
    \centering
    \begin{subfigure}{0.4\textwidth}
        \includegraphics[width=\linewidth]{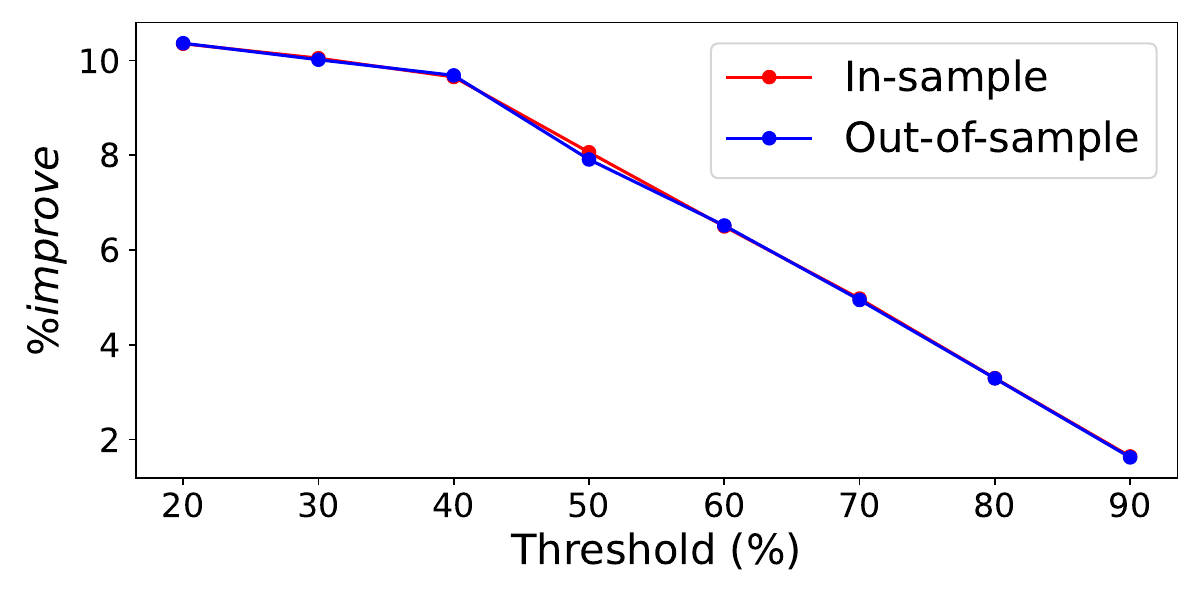}
        \caption{20 nodes}
    \end{subfigure}
    \begin{subfigure}{0.4\textwidth}
        \includegraphics[width=\linewidth]{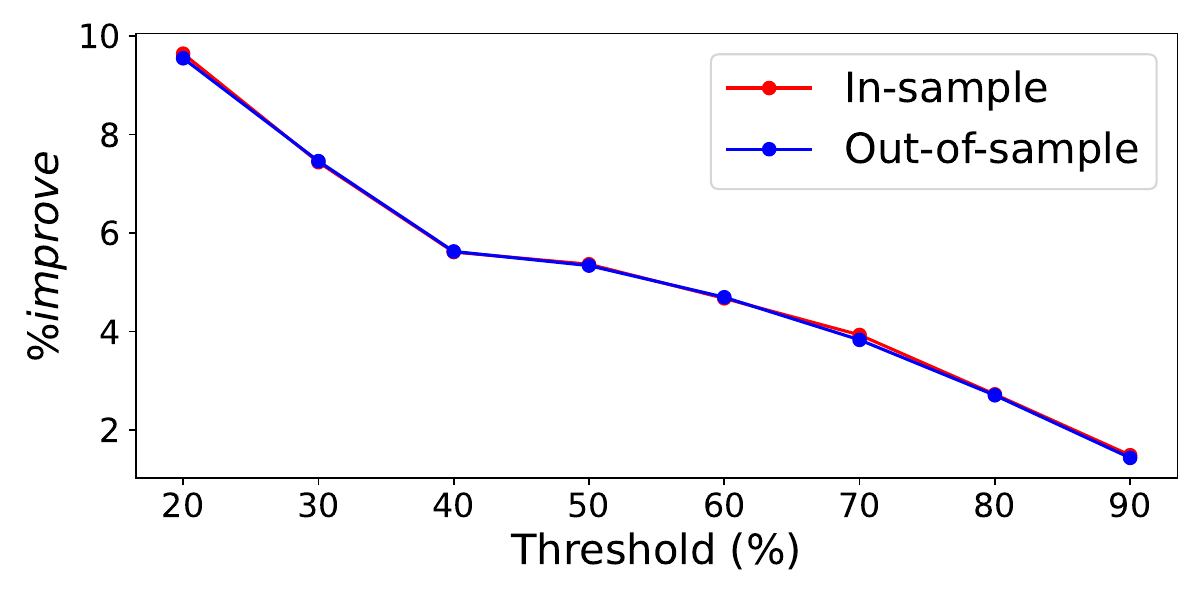}
        \caption{30 nodes}
    \end{subfigure}
    \begin{subfigure}{0.4\textwidth}
        \includegraphics[width=\linewidth]{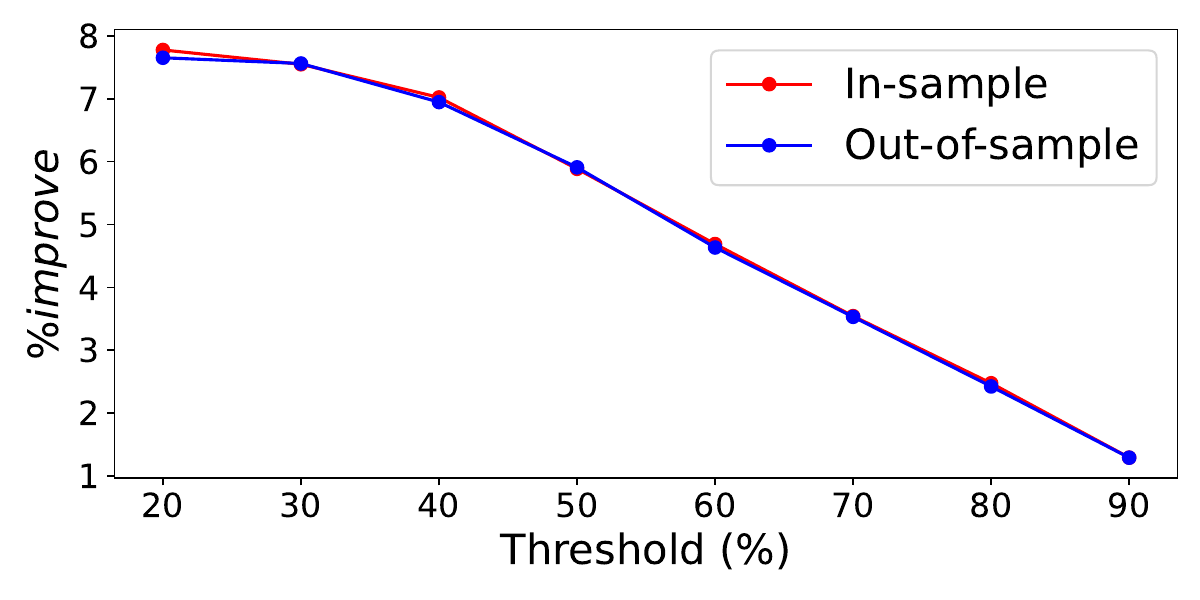}
        \caption{40 nodes}
    \end{subfigure}
    \begin{subfigure}{0.4\textwidth}
        \includegraphics[width=\linewidth]{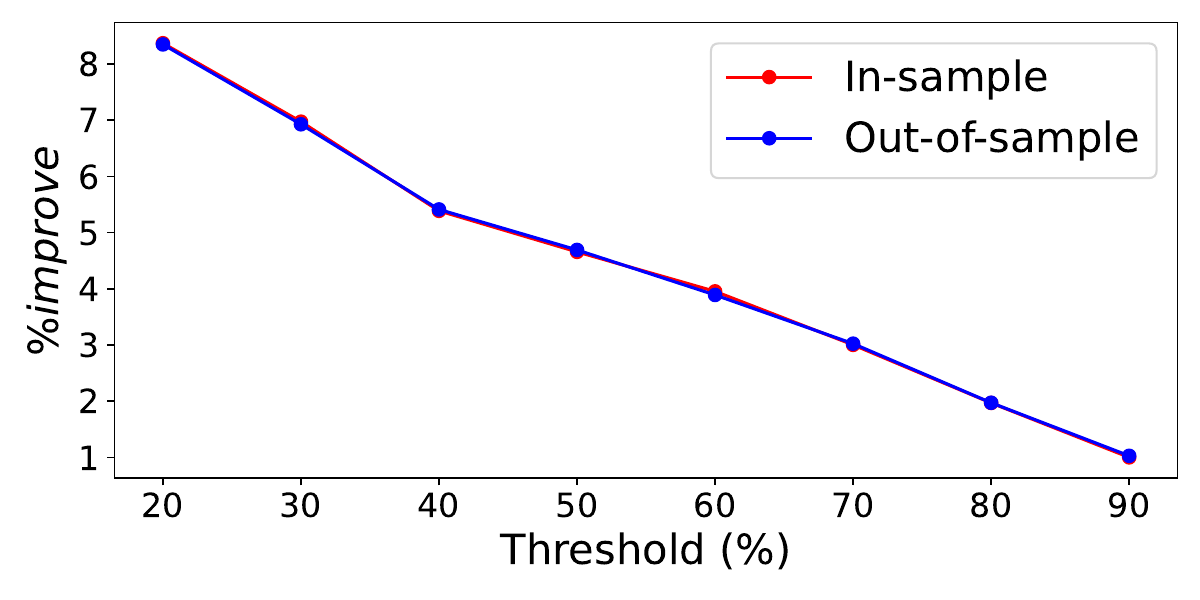}
        \caption{50 nodes}
    \end{subfigure}
    \caption{Estimation on synthetic observations generated by the CRL modes.} 
    \label{fig:u_constrained} 
\end{figure}

Taken together, these results demonstrate that the proposed CRL model offers estimation performance that is at least comparable to, and often better than, the standard RL model in the presence of upper-bound constraints. In particular, CRL  better aligns with realistic travel behavior when such feasibility constraints are binding.

\subsubsection{Route Choice for Rechargeable Vehicles}
We next conduct experiments in a more challenging setting that incorporates an 
\emph{energy consumption constraint} into the route choice behavior, mimicking 
the case of rechargeable vehicles. This setting is particularly relevant for 
modern transportation networks, where travelers must ensure that their remaining 
battery energy never falls below zero along the chosen route.  

\paragraph{Graph generation.}  
As in the previous case, we generate DAGs with sizes ranging from $20$ to $50$ 
nodes. To simulate the availability of recharging infrastructure, a set of 
``stations'' is randomly placed among the nodes of each graph. The number of 
stations is set to $10\%$ of the DAG size, reflecting the fact that only a 
limited fraction of nodes allow recharging.  

\paragraph{Observation generation.}  
Since the standard RL model cannot simulate the action of fully recharging 
energy at stations (and thus cannot generate feasible trajectories under the 
energy constraint), we rely exclusively on the proposed CRL model to generate 
observations in this setting. This highlights an important strength of the CRL 
framework: its ability to naturally incorporate feasibility constraints into 
the path choice process.  For each DAG size, we generate $5$ independent graph instances with randomly 
placed stations. For each instance, synthetic route choice data are generated 
under the CRL model.

\begin{figure}[htb] 
    \centering
    \begin{subfigure}{0.8\textwidth}
        \includegraphics[width=\linewidth]{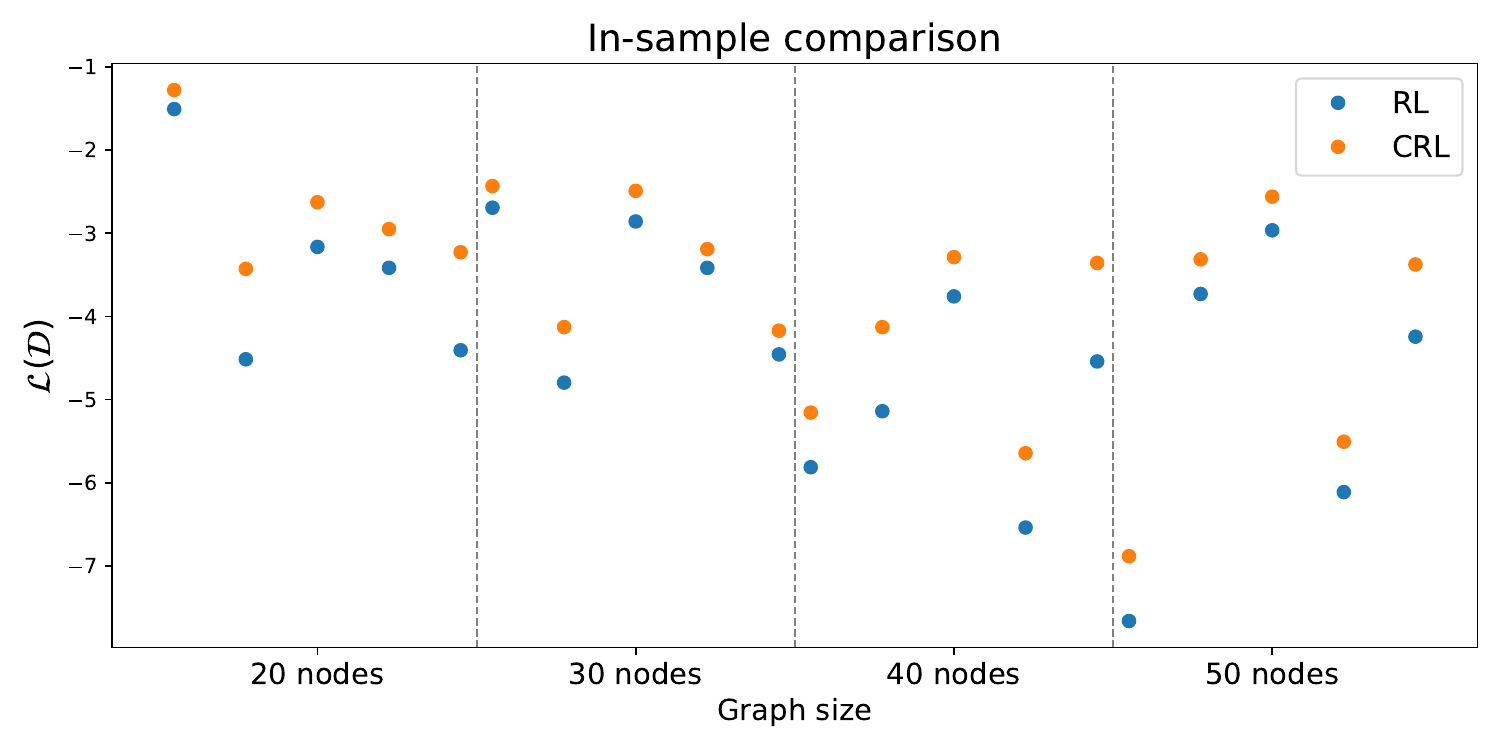}
        \caption{In-sample}
        \label{subfig:u_fuel_train}
    \end{subfigure}
    \begin{subfigure}{0.8\textwidth}
        \includegraphics[width=\linewidth]{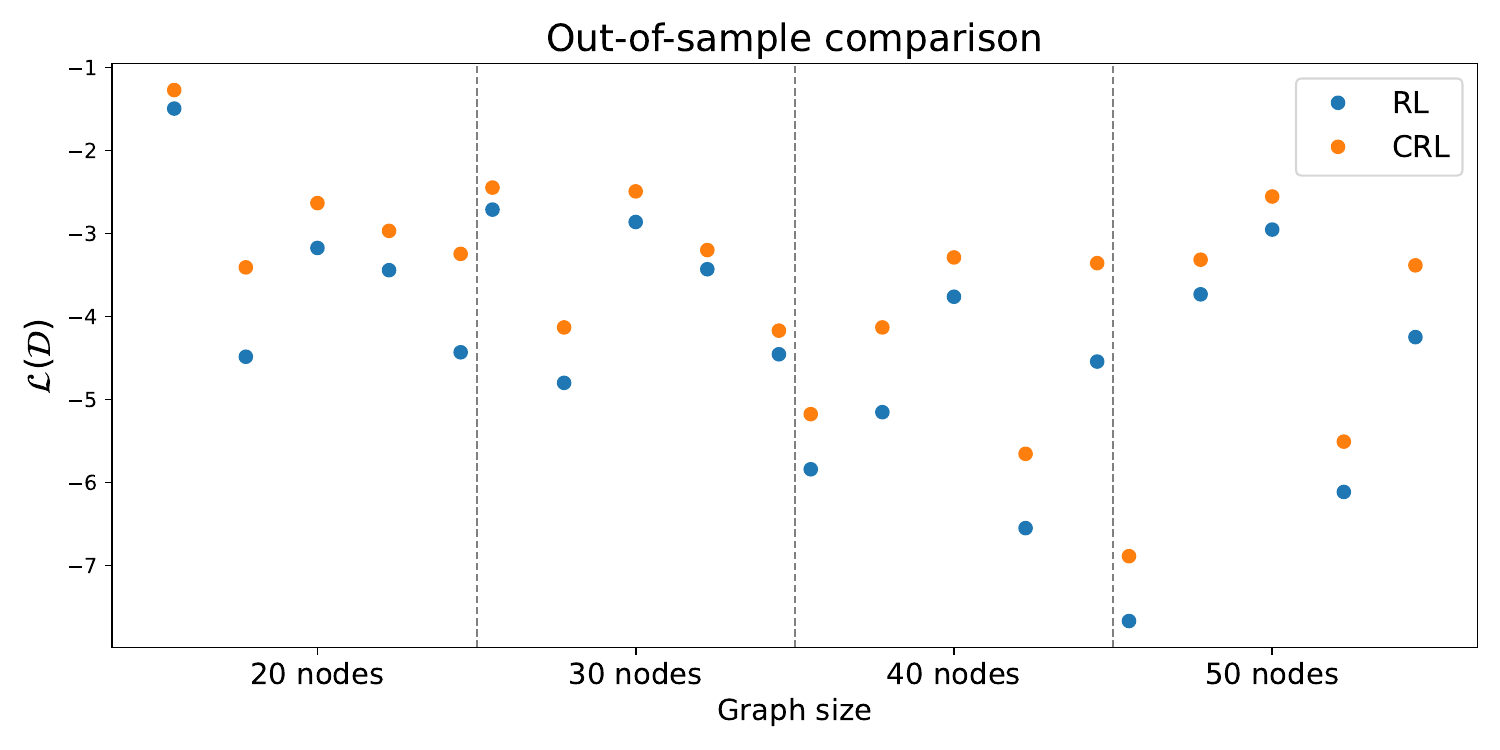}
        \caption{Out-of-sample}
        \label{subfig:u_fuel_test}
    \end{subfigure}
    \caption{In-sample  and out-of-sample average log-likelihood comparisons for rechargeable vehicles.} 
    \label{fig:u_fuel} 
\end{figure}

\autoref{fig:u_fuel} reports the estimation errors of both RL and CRL in the 
energy-constrained setting. The performance gap between the two models is 
substantially larger here: CRL (orange dots) consistently outperforms RL (blue dots) 
across all graph sizes. The average improvements of CRL over RL for DAGs with 
$20$, $30$, $40$, and $50$ nodes are $21\%$, $10\%$, $17\%$, and $13\%$, respectively.  
This significant gain highlights the importance of explicitly modeling feasibility 
constraints. While the upper-bound constraint considered previously only limits the 
set of acceptable routes by travel time, the recharging constraint introduces a 
more intricate form of path feasibility that requires accounting for energy dynamics 
along the route. Since the standard RL model cannot capture this complexity, it 
misallocates probability mass to infeasible trajectories, leading to poorer 
estimation. By contrast, the CRL model systematically enforces the energy 
constraint through the extended state space, thereby yielding consistently higher 
likelihood values.

\subsubsection{Estimation Time Comparison:}
As discussed earlier, the CRL model entails higher computational requirements 
compared to the standard RL model. We now briefly examine this aspect. 
Figure~\ref{fig:time-rl-generated} compares the estimation times of RL and CRL 
across different network sizes and threshold levels. For small networks 
(20--30 nodes), the estimation times of the two models are nearly identical, 
with both completing estimation in under two seconds across all threshold 
percentages. This suggests that the additional state-space expansion required 
by CRL does not impose a significant computational burden at modest network 
scales.

As the network size increases to 40 and 50 nodes, differences between the two 
models become more noticeable. While the RL model maintains nearly constant 
estimation times (around $0.4$--$0.6$ seconds), the CRL model exhibits increasing 
estimation times as the threshold percentage grows, reaching about $2.8$ seconds 
for the largest networks under the most relaxed constraints. This reflects the 
fact that CRL estimation must operate on an extended state space, whose size grows 
both with the number of nodes and with looser feasibility bounds.  

Overall, these results highlight a trade-off: CRL ensures behavioral consistency 
by eliminating infeasible paths, but at the expense of increased computational 
burden in large and cyclic networks.

\begin{figure}[H]
    \centering
    \includegraphics[width=0.9\textwidth]{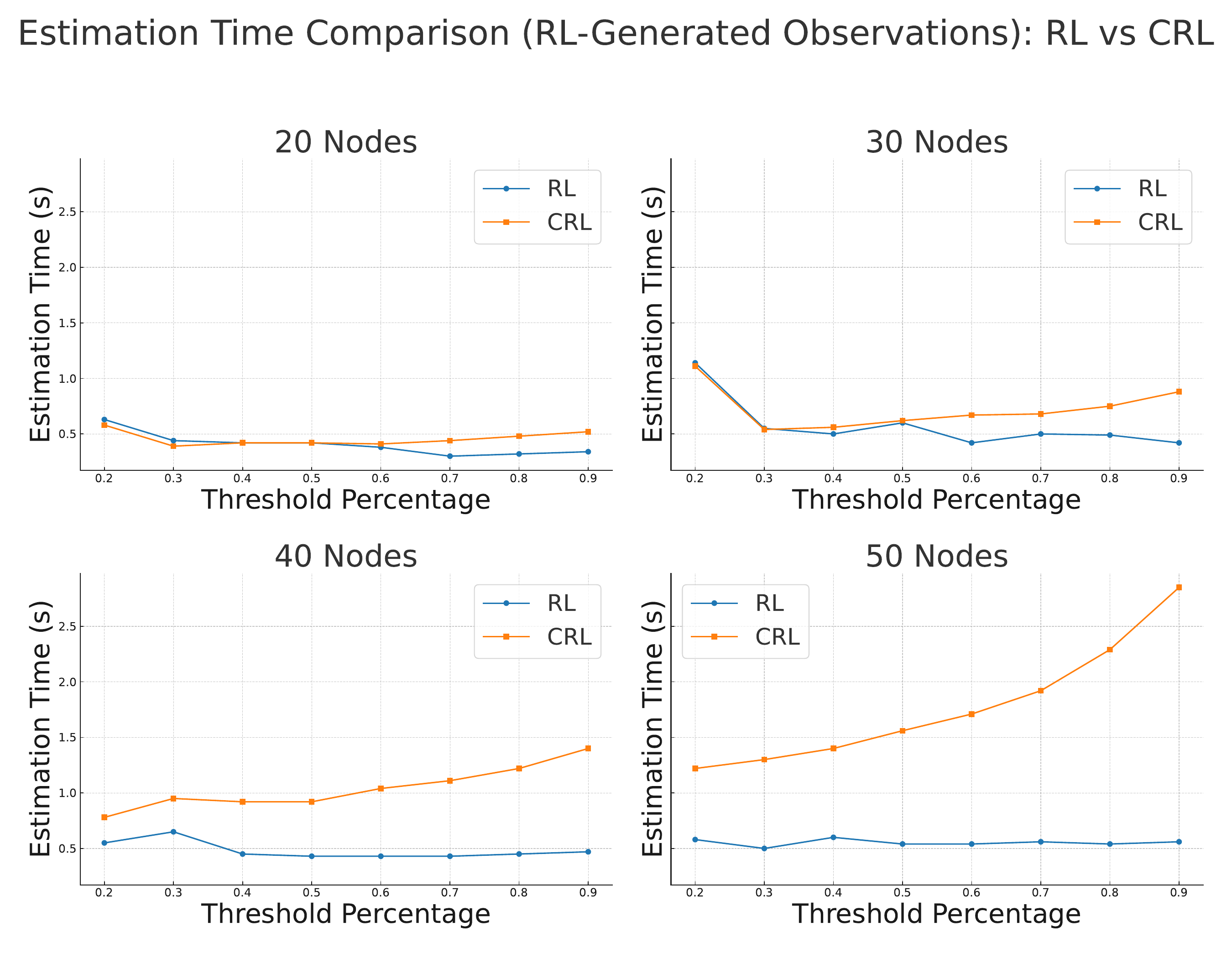}
    \caption{Estimation time comparison between RL and CRL on synthetic datasets across varying network sizes and threshold levels.}
    \label{fig:time-rl-generated}
\end{figure}

\subsection{Experiments on  Sioux Falls Network }

\paragraph{Sioux Falls network.}  
In addition to synthetic datasets, we evaluate the proposed CRL model on the 
well-known Sioux Falls transportation network. Originally introduced by 
\citet{LeBlanc1975SiouxFalls}, this network has become a standard benchmark in the 
transportation science and operations research literature for testing traffic 
assignment and route choice models. The Sioux Falls network consists of 
$24$ nodes and $76$ directed links, together with an origin--destination (OD) 
demand matrix that specifies travel demand between selected node pairs. Each link 
is associated with a free-flow travel time and a capacity parameter, allowing 
link travel times to be modeled under congestion effects (e.g., using the Bureau 
of Public Roads function).  

The Sioux Falls network is particularly attractive as a benchmark for route 
choice models because it strikes a balance between realism and tractability: 
it is large enough to capture meaningful network structure and route diversity, 
yet small enough to permit efficient computation of RL and 
CRL likelihoods. In our experiments, we use the Sioux 
Falls network to assess the practical applicability of the CRL model on a 
realistic network topology and to compare its performance against the standard 
RL model. 

\begin{figure}[htb] 
    \centering
    \includegraphics[width=0.5\linewidth]{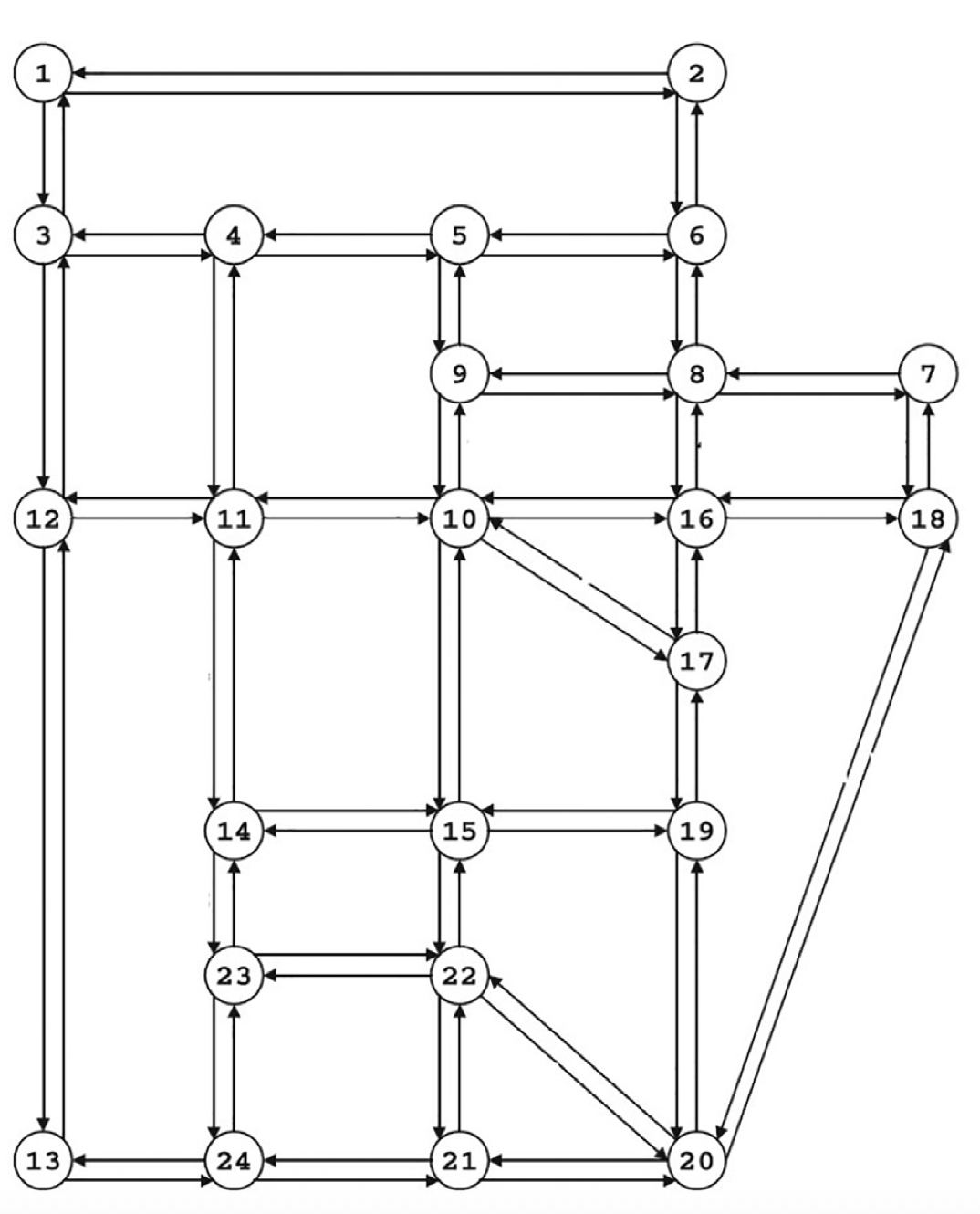}
    \caption{Sioux Falls network.} 
    \label{fig:siouxfalls} 
\end{figure}

As mentioned in Section~\ref{sec:extended-state-space}, one of the main advantages of the CRL model is that, 
under mild assumptions on the cost function, the Bellman equations always yield a 
unique solution. This solution can be computed efficiently either through value 
iterations or matrix inversion, thereby ensuring stable and efficient estimation.  
In contrast, as highlighted in \citep{MaiFrejinger22}, the original RL model often fails to 
produce reasonable parameter estimates, particularly in networks with many cycles. 
In such cases, the recursive structure of the Bellman equations may break down or 
lead to unstable likelihood evaluations. The Sioux Falls network provides a 
challenging testbed in this respect: its topology contains numerous cycles, making 
estimation under the standard RL model highly unstable.  

The primary objective of this experiment is therefore to compare the estimation 
stability of RL and CRL on the Sioux Falls network. By evaluating performance on 
this benchmark, we can assess whether the CRL formulation indeed provides the 
robustness required for realistic networks with cyclic structures.

\paragraph{Observation generation.}  
Due to the failure of the RL model in estimation (as reported below), we rely on the 
CRL model to generate observations. In this setting, the cost function is specified as 
$c(s,s') = 1$ for every edge $(s,s')$. This specification effectively imposes a 
constraint on the maximum number of edges that a traveler can traverse to reach the 
destination. Importantly, this choice satisfies the assumption stated in Section~\ref{sec:extended-state-space}, 
and therefore guarantees that the value function can always be computed successfully.  

Following the same experimental protocol as before, we generate $3000$ in-sample 
trajectories and $1000$ out-of-sample trajectories. The CRL model incorporates six 
attributes to simulate route choice probabilities: \emph{road capacity, road length, 
travel time, left turns, right turns}, and \emph{U-turns}. The associated parameter 
vector used for data generation is 
$\beta = [-0.5,\; 0,\; 1,\; -0.1,\; -0.05,\; -0.3]^{\top}.$
Note that since some of these coefficients are positive, certain arcs may yield 
positive utility values. This, in turn, can cause the unconstrained RL model to 
misbehave and fail during estimation, highlighting the necessity of the CRL 
framework for stable likelihood evaluation in such settings.

\begin{table}[htbp]
\centering
    \caption{Average log-likelihood values at estimation for the Sioux Falls dataset}
    \label{tab:siouxfalls}
\begin{tabular}{c|c|c}
\hline
$\alpha$ (upper bound on number of edges) & RL   & CRL  \\ \hline
10  & --   & -4.70  \\
15  & --   & -11.69 \\
20  & --   & -17.71 \\
25  & --   & -23.90 \\ \hline
\end{tabular}
\end{table}


The estimated average log-likelihood values for the two models are reported in 
\autoref{tab:siouxfalls}. As shown, the RL model fails to return reasonable 
log-likelihood values across all experimental setups. This failure occurs because 
the solver is unable to compute the Bellman equations successfully in networks 
with many cycles, resulting in undefined values. We denote these unsuccessful 
cases by “-” in the table.  In contrast, the CRL model can be estimated successfully in all settings. 
Furthermore, its estimation results exhibit a clear downward trend as the maximum 
number of edges allowed by the constraint $\alpha$ increases. When the choice set 
is very small (e.g., routes with no more than $10$ edges), the estimated 
log-likelihood values are relatively high. As $\alpha$ increases (e.g., allowing 
up to $25$ edges), the log-likelihood values decrease substantially. This behavior 
is intuitive: for a fixed number of observations, expanding the restricted 
universal choice set introduces many additional feasible paths, thereby lowering 
the probability assigned to any particular observed route.  

These results illustrate two key insights. First, the CRL formulation provides a 
stable and feasible estimation procedure even in challenging cyclic networks such 
as Sioux Falls, where the original RL model completely fails. Second, the 
relationship between $\alpha$ and the log-likelihood reflects the natural trade-off 
between model flexibility (larger choice sets) and estimation sharpness (higher 
likelihood concentration on observed data).

\subsection{Real-world Large-scale Traffic Network}

In this section, we present a comparative analysis of estimation results for the 
RL and CRL models. The dataset employed is the same as that used in previous 
studies on route choice modeling~\citep{FosgFrejKarl13,MaiFosFre15,MaiBasFre15_DeC}, collected in Borlänge, Sweden.  The underlying network consists of $3077$ nodes and $7459$ links. Since the network is uncongested, link travel times can be treated as static and 
deterministic. The sample comprises $1832$ observed trips, each corresponding to 
a simple path containing at least five links. In total, the dataset covers 
$466$ destinations, $1420$ distinct origin--destination (OD) pairs, and more 
than $37{,}000$ link choices. The instantaneous utility function is given as follows:
\begin{equation}\nonumber
    v(a|k;\beta) = \beta_{TT} TT(a) + \beta_{LT} LT(a|k) + \beta_{LC} LC + \beta_{UT} UT(a|k) 
\end{equation}
in which $TT(a)$ is the travel time on link $a$, $LT(a|k)$ is a left turn dummy of the turn from link $k$ to $a$, $LC$ is a constant one value and $UT(a|k)$ is a u-turn dummy from $k$ to $a$.


Both RL and CRL are applied for in-sample estimation on the observation dataset. 
All experimental settings follow those in the original RL model publication 
\citep{FosgFrejKarl13}. In this experiment, the constraint incorporated by the 
CRL model is a bound constraint on the \emph{link count} of each route, which is 
defined as the sum of the $LC$ attribute values over the links in the route.   The distribution of link counts across the $1832$ observed routes is illustrated 
in \autoref{fig:borlange-counts}. The maximum number of links in a single route is 
$106$, which naturally serves as the upper bound for the constraint. 
It is important to note that link count is already a discrete variable, making it 
straightforward to apply the CRL model in this context. Moreover, as discussed in 
Section~\ref{sec:extended-state-space}, the strict positivity of the cost values (here given by $LC$) ensures 
that the extended network is cycle-free. This property guarantees that the CRL 
model can be estimated efficiently using standard solution methods.

\begin{figure}[htbp] 
    \centering
    \includegraphics[width=0.8\linewidth]{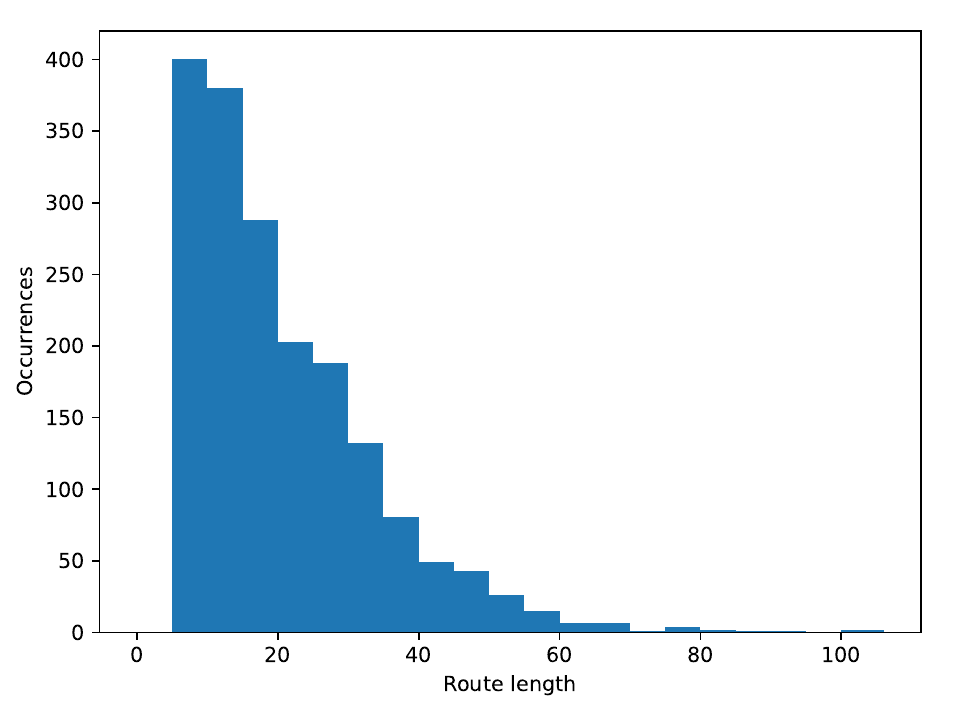}
    \caption{Distribution of link counts in the observation set.} 
    \label{fig:borlange-counts} 
\end{figure}

\begin{table}[htbp]
    \centering
\begin{tabular}{l|l|rrrr|c}
\multirow{2}{*}{Model} & \multirow{2}{*}{Property}  & \multicolumn{4}{c|}{Attributes}   & \multirow{2}{*}{$\cL(\mathcal{D})$} \\ \cline{3-6}
                       &                            & $TT$   & $LC$   & $LT$   & $UT$   &                                   \\ \hline
\multirow{3}{*}{RL}    & $\widehat{\mathbf{\beta}}$ & -2.49  & -0.41  & -0.93  & -4.46  & \multirow{3}{*}{-3.44}             \\
                       & Std. Err.                  & 0.07   & 0.01   & 0.03   & 0.10   &                                   \\
                       & $t$-test(0)                & -37.90 & -39.58 & -37.25 & -45.39 &                                   \\ \hline
\multirow{3}{*}{CRL}   & $\widehat{\mathbf{\beta}}$ & -2.51  & -0.41  & -0.93  & -4.56  & \multirow{3}{*}{-3.39}             \\
                       & Std. Err.                  & 0.07   & 0.01   & 0.03   & 0.10   &                                   \\
                       & $t$-test(0)                & -38.48 & -37.38 & -34.11 & -45.05 &                                   
\end{tabular}
    \caption{Estimation results on the Borlänge dataset.}
    \label{tab:borlange}
\end{table}

The estimation results for both RL and CRL are reported in \autoref{tab:borlange}. 
Overall, the estimated parameters are very close across the two models, suggesting 
that both approaches capture the main behavioral patterns in the dataset. However, 
the CRL model achieves a slightly higher log-likelihood value ($-3.39$ compared to $-3.44$ 
for RL), corresponding to an improvement of approximately $1.5\%$ in the likelihood 
value.

We also apply a cross-validation approach to compare the prediction performance of RL and CRL. The full sample set is randomly split into two subsets: 80\% of the observations are used for estimation, while the remaining 20\% serve as a holdout set to assess predictive capability. In total, 30 different estimation–holdout partitions are generated. For each partition, we estimate the parameters on the estimation set and then compute the average log-likelihood on the holdout set using the estimated parameters. These predicted log-likelihood values across the 30 samples are shown in \autoref{fig:borlange-out-samples}, with the samples arranged in ascending order of RL’s prediction values for clarity. CRL consistently outperforms RL across all samples, with an average predicted log-likelihood of $-3.38$ compared to $-3.43$ for RL.
\begin{figure}[htb]
    \centering
    \includegraphics[width=0.55\textwidth]{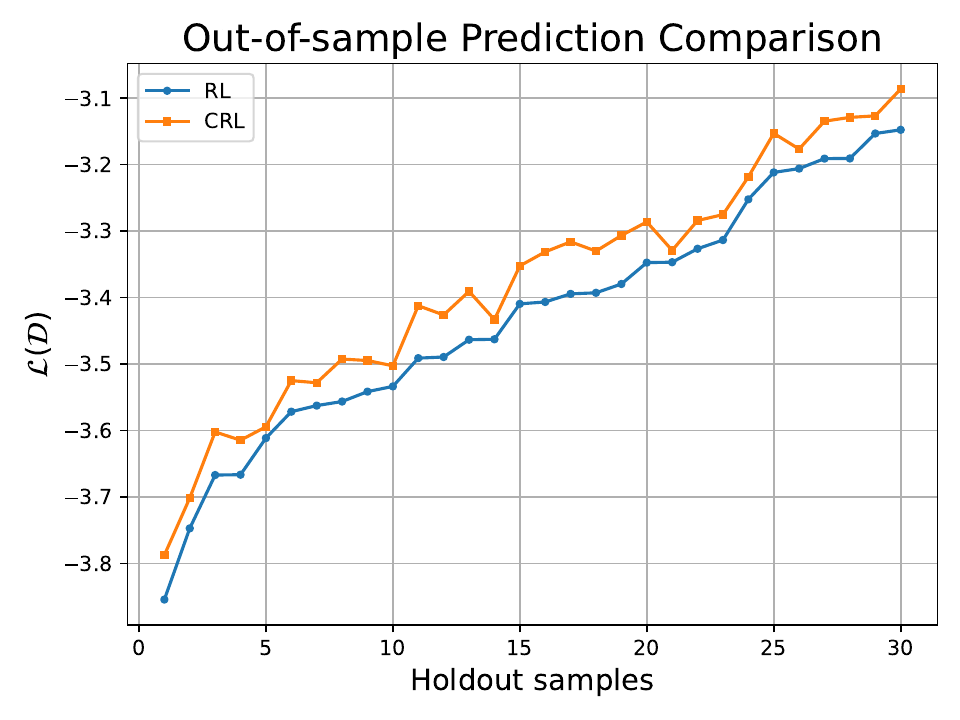}
    \caption{Average log-likelihood values over holdout samples for the Borlänge dataset.}
    \label{fig:borlange-out-samples}
\end{figure}
 
This modest yet systematic improvement demonstrates that the CRL formulation is 
capable of capturing additional structural constraints in traveler behavior that 
are overlooked by the unconstrained RL model. It should be noted, however, that 
the magnitude of the improvement is relatively smaller compared to the case of 
synthetic data reported earlier. This is largely due to the fact that the Borlänge 
network is uncongested, so travelers are not strongly constrained by travel time 
or link count. As a result, paths with the shortest travel times dominate the 
choice probabilities, leaving limited scope for CRL to provide large gains over 
the unconstrained RL formulation. Nevertheless, even under such favorable 
conditions for RL, the CRL model still delivers consistent in-sample and 
out-of-sample improvements, underscoring its robustness in capturing feasibility 
conditions whenever they are present.

It is worth noting that the computational cost of CRL is substantially higher 
than that of RL (approximately $50\times$ higher for this real dataset). This 
increase reflects the added complexity of operating on the extended state space. 
Nevertheless, the improved model fit demonstrates that CRL provides significant 
value in situations where datasets involve constraints or atypical features that 
cannot be adequately captured by the standard RL model. Importantly, estimation 
remains feasible in practice: it can be completed in under ten hours, and further 
speedups are possible through parallel computing or decomposition techniques 
\citep{MaiBasFre15_DeC}. These observations highlight the practical viability of 
CRL for large-scale, real-world datasets.

\section{Conclusion}\label{sec:concl}
This paper introduced the CRL model, which extends 
the classical RL framework by explicitly incorporating feasibility constraints 
into the universal choice set. By enforcing that any path violating a constraint 
receives zero probability, the CRL model yields behaviorally consistent 
predictions while retaining the computational advantages of RL. We showed that 
estimation can be made tractable through a state-space extension, established 
equivalence to path-based MNL under restricted universal choice sets, and demonstrated 
extensions to multiple constraints and nested RL. Empirical results on both 
synthetic and real networks confirm that CRL provides more realistic predictions 
and stable estimation compared to RL, especially in cyclic networks. Future work 
may focus on improving computational efficiency of the extended state-space 
estimation for large-scale networks, as well as exploring relaxed formulations 
where infeasible paths are penalized rather than strictly excluded, thereby 
broadening the applicability of CRL in practical contexts.

\bibliographystyle{plainnat_custom}
\bibliography{refs}

\appendix
\clearpage 

\begin{center}
    {\Huge Appendix}
\end{center}

\section{Extension to Multiple-Constrained RL}
We now extend the CRL framework to accommodate multiple constraints. Such a 
multi-constrained formulation generalizes the single-constraint CRL model to 
settings where several feasibility conditions must be satisfied simultaneously. 
For instance, an electric vehicle traveler may face both a travel-time deadline 
and a battery capacity limit, while a bike-sharing trip may be restricted by a 
maximum rental duration and station availability. The multi-constraint CRL 
framework incorporates these conditions directly into the recursive structure, 
preserving the Markovian property while systematically assigning zero 
probability to any path that violates one or more constraints.

As a starting point, we provide a formal definition of a feasible 
path in the context of multiple constraints.

\begin{definition}[Feasible Path Set under Multiple Stepwise Constraints]
Let $\mathbf{c}(s,s') = \bigl(c^1(s,s'), \ldots, c^K(s,s')\bigr)$ denote a 
$K$-dimensional vector of transition costs, and let 
$\boldsymbol{\alpha} = (\alpha_1,\ldots,\alpha_K)$ be the corresponding vector of 
upper bounds. For a path $\tau = \{s_0,\ldots,s_T\}$, define the accumulated 
cost vector up to step $t$ as
\[
\mathbf{C}(S_t) = \left( \sum_{j=0}^{t-1} c^1(s_j,s_{j+1}), ~ \ldots, ~
                     \sum_{j=0}^{t-1} c^K(s_j,s_{j+1}) \right),
\]
where $S_t = \{s_0,\ldots,s_t\}$ is the partial trajectory up to time $t$.  

The \emph{feasible path set} under multiple stepwise constraints is defined as
\[
\widetilde{\Omega}^{\boldsymbol{\alpha}} = 
\left\{ \tau = \{s_0,\ldots,s_T\} ~\middle|~ s_T = d,~
C^k(S_t) \leq \alpha_k \;\; \forall k=1,\ldots,K,~\forall t=0,\ldots,T \right\}.
\]
\end{definition}
In words, a path $\tau$ is feasible if, at every intermediate step, the 
accumulated cost in \emph{each} constraint dimension $k$ does not exceed its 
respective bound $\alpha_k$. Any path that violates even a single constraint at 
any step is excluded from $\widetilde{\Omega}^{\boldsymbol{\alpha}}$ and is 
therefore assigned zero probability in the CRL model. This ensures that 
constraint feasibility is enforced throughout the trajectory, not merely at 
the destination.

The state-space extension framework introduced above can be naturally generalized 
to handle multiple constraints simultaneously.
 Specifically, let 
\[
\mathbf{c}(s,s') = \bigl(c^1(s,s'), c^2(s,s'), \ldots, c^K(s,s')\bigr)
\]
denote a $K$-dimensional vector of constraint functions associated with the 
transition from $s$ to $s'$, where each $c^k(s,s')$ represents the cost 
contributing to the $k$-th constraint (e.g., travel time, energy consumption, 
rental duration). Let $\boldsymbol{\alpha} = (\alpha_1, \alpha_2, \ldots, \alpha_K)$ 
be the corresponding vector of upper bounds.  

\paragraph{Extended state space.}  
We define an extended state space as 
\[
\widetilde{\mathcal{S}} = \{ (s,\mathbf{z}) \;\mid\; s \in \mathcal{S},\; \mathbf{z} \in \mathbb{R}^K \},
\]
where $\mathbf{z} = (z^1,\ldots,z^K)$ records the accumulated values of all 
constraints along the path. For an extended state $\widetilde{s} = (s,\mathbf{z})$, 
the set of feasible successors is defined as
\[
\widetilde{N}(\widetilde{s}) 
= \Bigl\{ (s',\mathbf{z}') \;\Big|\; s' \in N(s),\; 
    \mathbf{z}' = \mathbf{z} + \mathbf{c}(s,s') \Bigr\}.
\]

\paragraph{Value function recursion.}  
The recursive definition of the value function now reads
\[
\widetilde{V}(\widetilde{s}) =
\begin{cases}
    0, & \text{if } s = d \;\;\text{and}\;\; z^k \leq \alpha_k \;\; \forall k, \\[6pt]
    -\infty, & \text{if } \exists k: z^k > \alpha_k, \\[6pt]
    \mu \ln \left( 
        \sum_{\widetilde{s}' \in \widetilde{N}(\widetilde{s})} 
        \exp\!\left( \tfrac{1}{\mu} \bigl[ 
            \widetilde{v}(\widetilde{s}' \mid \widetilde{s}) 
            + \widetilde{V}(\widetilde{s}') 
        \bigr] \right) 
    \right), & \text{otherwise}.
\end{cases}
\]

\paragraph{Choice probabilities.}  
The conditional choice probabilities at each extended state follow the standard 
MNL form:
\[
\widetilde{\pi}(\widetilde{s}' \mid \widetilde{s}) 
= \frac{\exp\!\left( \tfrac{1}{\mu} \bigl[ 
    \widetilde{v}(\widetilde{s}' \mid \widetilde{s}) 
    + \widetilde{V}(\widetilde{s}') \bigr]\right)}
{\sum\limits_{\bar{s} \in \widetilde{N}(\widetilde{s})} 
\exp\!\left( \tfrac{1}{\mu} \bigl[ 
    \widetilde{v}(\bar{s} \mid \widetilde{s}) 
    + \widetilde{V}(\bar{s}) \bigr]\right)}.
\]
As in the single-constraint case, the path choice probabilities generated by the 
multi-constrained CRL model are equivalent to those of a multinomial logit (MNL) 
model defined over the restricted feasible path set 
$\widetilde{\Omega}^{\boldsymbol{\alpha}}$. Specifically, for any path 
$\tau$ we have
\begin{equation}
P^{\text{CRL}}(\tau) = 
\begin{cases}
\dfrac{\exp\!\left( \tfrac{1}{\mu} v(\tau) \right)}
      {\sum\limits_{\tau' \in \widetilde{\Omega}^{\boldsymbol{\alpha}}} 
       \exp\!\left( \tfrac{1}{\mu} v(\tau') \right)} 
& \text{if } \tau \in \widetilde{\Omega}^{\boldsymbol{\alpha}}, \\[12pt]
0 & \text{otherwise},
\end{cases}
\end{equation}
where $v(\tau)$ denotes the deterministic utility of path $\tau$.

\paragraph{Computational complexity.} An important consideration in the multi-constraint CRL framework is the 
computational complexity of estimation, which grows with the number of 
constraints. Since the extended state space augments each network state $s \in 
\mathcal{S}$ with a $K$-dimensional accumulated cost vector $\mathbf{z} \in 
\mathbb{R}^K$, the total number of extended states is proportional to 
$|\mathcal{S}| \times |\mathcal{Z}_1| \times \cdots \times |\mathcal{Z}_K|$, 
where $\mathcal{Z}_k$ denotes the discretized support of the $k$-th constraint. 
Thus, the size of the state space grows exponentially in $K$, leading to a 
potential \emph{curse of dimensionality}. In practice, however, this issue can be 
mitigated in several ways: many constraints (e.g., energy consumption, rental 
duration) are naturally discretized, and the feasible region defined by the 
bounds $\alpha_k$ significantly restricts the effective state space. Moreover, 
infeasible states are pruned dynamically during value function iteration, which 
reduces the computational burden relative to the worst-case size of the 
augmented state space. Nevertheless, when the number of constraints becomes 
large, estimation of the CRL model may become challenging, and efficient 
approximation strategies or problem-specific relaxations may be required.

\section{Constrained Nested RL (CNRL) Model}
We now discuss how our framework can be extended to the nested recursive logit 
(NRL) model \citep{MaiFosFre15}. The NRL model is designed to capture 
correlations between path utilities by introducing a nesting structure in the 
error terms. In the constrained setting, this nesting can be incorporated into 
the CRL framework by allowing the scale parameter $\mu$ to vary across states. 
Estimation can then be carried out on the extended state space using the NFXP 
algorithm, analogous to the procedure described for CRL.  

In particular, each extended state $\widetilde{s} \in \widetilde{\mathcal{S}}$ 
is associated with a local dispersion parameter $\mu_{\widetilde{s}}$, which 
governs the variance of the random utility shocks at that state. This 
state-dependent parameterization provides the flexibility to represent 
correlation patterns across subsets of alternatives, while the constrained 
structure of CRL continues to ensure that infeasible paths are excluded from the 
universal choice set.

\paragraph{Value function recursion.}  
The recursive definition of the value function in the extended state space now 
takes the form
\[
\frac{1}{\mu_{\widetilde{s}}}\widetilde{V}(\widetilde{s}) =
\begin{cases}
    0, & \text{if } \mathcal{N}(\widetilde{s}) = d \text{ and } \mathcal{C}(\widetilde{s}) \leq \alpha, \\[6pt]
    -\infty, & \text{if } \mathcal{C}(\widetilde{s}) > \alpha, \\[6pt]
    \ln \left( 
        \sum_{\widetilde{s}' \in \widetilde{N}(\widetilde{s})} 
        \exp\!\left( \tfrac{1}{\mu_{\widetilde{s}}} \Bigl[ 
            \widetilde{v}(\widetilde{s}' \mid \widetilde{s}) 
            + \widetilde{V}(\widetilde{s}') 
        \Bigr] \right) 
    \right), & \text{otherwise}.
\end{cases}
\]

\paragraph{Choice probabilities.}  
The conditional choice probabilities follow directly as
\[
\widetilde{\pi}(\widetilde{s}' \mid \widetilde{s}) 
= \frac{\exp\!\left( \tfrac{1}{\mu_{\widetilde{s}}} \bigl[ 
    \widetilde{v}(\widetilde{s}' \mid \widetilde{s}) + \widetilde{V}(\widetilde{s}') \bigr]\right)}
{\sum\limits_{\bar{s} \in \widetilde{N}(\widetilde{s})} 
\exp\!\left( \tfrac{1}{\mu_{\widetilde{s}}} \bigl[ 
    \widetilde{v}(\bar{s} \mid \widetilde{s}) + \widetilde{V}(\bar{s}) \bigr]\right)}.
\]
Similar to the CRL formulation, 
any path that violates one or more feasibility constraints is systematically 
assigned zero choice probability. This is ensured by the recursive structure of 
the value function on the extended state space: whenever the accumulated cost 
exceeds the threshold in any constraint dimension, the corresponding extended 
state is assigned a value of $-\infty$. Consequently, all paths passing through 
such infeasible states are excluded from the feasible choice set. This guarantees 
that the probability distribution generated by the CNRL model is supported only 
on feasible paths, thereby preserving behavioral consistency while still 
capturing correlation patterns through the nesting structure.

\paragraph{Estimation via NFXP.}  
As in the CRL model, estimation of the constrained nested RL (CNRL) model can be carried out using the 
Nested Fixed-Point (NFXP) algorithm \citep{Rust87}. In the inner loop, the state-dependent value 
functions $\widetilde{V}(\widetilde{s})$ are computed by value iteration on the 
extended state space, using the recursion above. In the outer loop, the 
likelihood function is maximized with respect to the structural parameters 
$(\beta, \{\mu_{\widetilde{s}}\}_{\widetilde{s} \in \widetilde{\mathcal{S}}})$.  

\paragraph{Computational considerations.}  
Compared to CRL, the estimation of CNRL is more computationally demanding due to 
the heterogeneity of $\mu_{\widetilde{s}}$. In particular, the recursion must be 
solved with state-specific scaling, which precludes the simplification of 
rewriting the Bellman equations in linear form. Nevertheless, feasibility 
constraints imposed by the CRL framework still ensure that the extended state 
space is finite (under the assumption of discrete costs), and that infeasible 
states are pruned dynamically. This guarantees that the value iteration procedure 
converges and that the likelihood is well-defined.

\end{document}